\documentclass[runningheads]{llncs}
\usepackage{graphicx}
\usepackage{newtxtext,newtxmath}      
\usepackage{amsmath,url}

\usepackage{textgreek}
\usepackage{wrapfig}
\usepackage[table]{xcolor}

\newcommand{\bGamma}{\mathbf{\Gamma}}
\newcommand{\bR}{\mathbf{R}}
\newcommand{\CSP}[1]{\ensuremath{\operatorname{CSP}(#1)}}
\newcommand{\bv}{\mathbf{v}}
\DeclareMathOperator{\ar}{ar}
\DeclareMathOperator{\pol}{Pol}
\DeclareMathOperator{\mpol}{MPol}

\def\house{\widehat{\Join}}

\newif\ifTR\TRtrue
\newif\ifNO\NOtrue
\newif\ifCOM\COMfalse

\begin{document}

\title{On the induced problem for fixed-template CSPs}
%
%
\author{Rustem Takhanov\inst{1}\orcidID{0000-0001-7405-8254}}
\authorrunning{R. Takhanov}
%
\institute{
Nazarbayev University, 53 Kabanbay Batyr Ave, Astana, Kazakhstan
\email{rustem.takhanov@nu.edu.kz}}
\maketitle              
\begin{abstract}
The Constraint Satisfaction Problem (CSP) is a problem of computing a homomorphism $\bR\to \bGamma$ between two relational structures, where $\bR$ is defined over a domain $V$ and $\bGamma$ is defined over a domain $D$. 
In a fixed template CSP, denoted $\textsc{CSP}(\bGamma)$, the right side structure $\bGamma$ is
fixed and the left side structure $\bR$ is unconstrained. 
In the last two decades it was discovered that the reasons that make fixed template CSPs polynomially solvable are of algebraic nature, namely, templates that are tractable should be preserved under certain polymorphisms. 
From this perspective the following problem looks natural: given a prespecified finite set of algebras ${\mathcal B}$ whose domain is $D$, is it possible to present the solution set of a given instance of $\textsc{CSP}(\bGamma)$ as a subalgebra of ${\mathbb A}_1\times ... \times {\mathbb A}_{|V|}$ where ${\mathbb A}_i\in {\mathcal B}$?

We study this problem and show that it can be reformulated as an instance of a certain fixed-template CSP over another template $\bGamma^{\mathcal B}$. 

We study conditions under which $\textsc{CSP}(\bGamma)$ can be reduced to $\textsc{CSP}(\bGamma^{\mathcal B})$. 
This issue is connected with the so-called CSP with an input prototype, formulated in the following way: given a homomorphism from $\bR$ to $\bGamma^{\mathcal B}$ find a homomorphism from $\bR$ to $\bGamma$. We prove that if ${\mathcal B}$ contains only tractable algebras, then the latter CSP with an input prototype is tractable.
We also prove that $\textsc{CSP}(\bGamma^{\mathcal B})$ can be reduced to $\textsc{CSP}(\bGamma)$ if the set ${\mathcal B}$, treated as a relation over $D$, can be expressed as a primitive positive formula over $\bGamma$. 
\end{abstract}

\section{Introduction}\label{sec:intro}

The {\it constraint satisfaction problem (CSP)} can be formalized in the variable-value form as a problem of finding an assignment of values to a given set of variables, subject to constraints. There is also an equivalent formulation of it as a problem of finding a homomorphism $h: \bR\to \bGamma$ for two given finite relational structures $\bR$ and $\bGamma$. 

A special case, when the second relational structure is some fixed $\bGamma$ (called a template) and the domain of $\bGamma$ is Boolean, was historically one of the first NP-hard problems~\cite{Cook:1971}, and its study has attracted some attention since the 70s~\cite{Garey:1990,Schaefer:1978}. 
Feder and Vardi~\cite{feder98:monotone} formulated the so-called dichotomy conjecture that states that for a template $\bGamma$ over an arbitrary finite domain, $\CSP\bGamma$ is either polynomially solvable or NP-hard.
In \cite{Jeavons:1998} Jeavons showed that the complexity of $\CSP\bGamma$ is determined by the so-called polymorphisms of $\bGamma$. A polymorphism of a relation $\varrho\subseteq D^n$ is defined as an $m$-ary function $p: D^m\to D$ such that $\varrho$ is closed under the operation $p$ that is applied component-wise to tuples from $\varrho$. A polymorphism of a template $\bGamma$ is defined as a function that is a polymorphism of all relations in $\bGamma$. Jeavons's result implies that if two languages $\bGamma$ and $\bGamma'$ have the same polymorphisms, then $\CSP\bGamma$ and $\textsc{CSP}(\bGamma')$ are log-space inter-reducible.
Further development of this method~\cite{bulatov05:classifying,MarotiMcKenzie,Siggers,4ary} made it possible to
precisely delineate the borderline between the polynomial-time and NP-complete templates of the CSP. 
Bulatov~\cite{Bulatov17a} and Zhuk~\cite{Zhuk17} independently confirmed the Feder-Vardi conjecture. 

Polymorphisms are interesting objects by themselves. Obviously, if a template $\bGamma$ has a polymorphism $p$, then for any instance $\bR$ the set of solutions is preserved by $p$. In other words, any such polymorphism induces an algebra on the solution set. Is there any simple generalization of that algebra that also takes into account the structure $\bR$?





Inspired by this observation, we suggest generalizing algebras induced by polymorphisms of $\bGamma$ in the following way. Suppose that we have a finite set  ${\mathcal B}$ of similar algebras, i.e. every element ${\mathbb A}\in {\mathcal B}$ is a tuple $(D, o^{\mathbb A}_1, ..., o^{\mathbb A}_k)$ where $o^{\mathbb A}_i:D^{n_i}\to D, i\in 1,\dots, k$. Now, any instance of $\CSP\bGamma$ corresponds to another search problem, which we call {\em the induced problem}. 
More
precisely, let $(V,C)$ be an instance of $\CSP\bGamma$, where $V$ is a set of variables and $C$ is a set of constraints.
The goal of the induced problem is to assign an algebra ${\mathbb A}_v\in {\mathcal B}$ to each variable $v\in V$ in such a way that for every constraint
$\langle (v_1,...,v_k),\varrho\rangle\in C$, the relation $\varrho$ is a subalgebra of ${\mathbb A}_{v_1}\times ... \times {\mathbb A}_{v_p}$. If we are able to find such an assignment, it can be proved that the set of all solutions of the initial CSP instance (which is a subset of $D^V$) is a subalgebra of $\mathop\prod\limits_{v\in V}{\mathbb A}_{v}$. 

\vspace{-10pt}
\paragraph{\bf Motivation.} 
Green and Cohen~\cite{GREEN20081094} considered the following computational problem which is a special case of the induced problem. Given a CSP over $D=\{1,\cdots, d\}$, they formulate a search problem: for any variable $v\in V$ find a permutation $\pi_v: D\to D$ such that if an assignment to $(v_1, ..., v_p)$ is constrained to be in $\varrho\subseteq D^{p}$ (in the initial CSP), then the resultant permuted relation $\varrho' = \{(\pi_{v_1}(x_1), \cdots, \pi_{v_p}(x_p)) | (x_1, \cdots, x_p)\in \varrho\}$ is max-closed (i.e. has a polymorphism $\max(x, y)$). 

Indeed, suppose that we solved the described problem. Then, the initial CSP can be modified by a simple substitute $\varrho\rightarrow\varrho'$ in every constraint. This new CSP instance is called the permuted CSP.
Then to any solution $h:V\to D$ of the permuted CSP (whose existence can be identified efficiently because the CSP is max-closed), we correspond a solution of the initial one by $h'(v) = \pi_v^{-1}(h(v))$. 

To find such permutations $\pi_v, v\in V$ is computationally difficult in general. 
Green and Cohen's construction is a special case of ours, if we define a set of algebras $${\mathcal B} = \{{\mathbb A}_{\pi}| \pi {\rm \,\, is \,\, a\,\, permutation\,\,of\,\,}D\},$$ where for any permutation $\pi$ of $D$ we set $${\mathbb A}_{\pi} = \big(D, \max{}_{\pi}\big), \max{}_{\pi}(x,y) = \pi^{-1}\big(\max(\pi(x), \pi(y))\big).$$
The assignment of variables $v\to \pi_v$ corresponds to the assignment $v\to {\mathbb A}_{\pi_v}$ in our framework. 

\vspace{-10pt}
\paragraph{\bf Our results.} 
First we prove that a search for an assignment $v\to {\mathbb A}_{v}$ in the induced problem is equivalent to solving another fixed-template $\textsc{CSP}(\bGamma^\mathcal{B})$ and study the relationship between  $\textsc{CSP}(\bGamma)$ and $\textsc{CSP}(\bGamma^\mathcal{B})$. 
We prove that if the family $\mathcal{B}$ is tractable, i.e. all algebras in $\mathcal{B}$ are tractable (individually), and we are given a homomorphism $\bR\to \bGamma^{\mathcal{B}}$, then a homomorphism $h:\bR\to \bGamma$ can be found efficiently (Theorem~\ref{Btract}). This result generalizes the desirable property of Green and Cohen's construction, in which, given appropriate permutations of domains, the initial CSP can be efficiently solved.
As a corollary we obtain that if $\bGamma$  maps homomorphically to  $\bGamma^\mathcal{B}$, then $\textsc{CSP}(\bGamma)$ is reducible to $\textsc{CSP}(\bGamma^\mathcal{B})$ (Theorem~\ref{mainreduction1}).
Further we prove that if the family $\mathcal{B}$ has a certain structure, namely, when a certain relation defined by $\mathcal{B}$ (called the trace of $\mathcal{B}$) is expressible as a primitive positive formula over $\bGamma$, 
$\textsc{CSP}(\bGamma^\mathcal{B})$ can be reduced to $\textsc{CSP}(\bGamma)$ (Theorems~\ref{simple-red},~\ref{strong-red} and Corollary~\ref{BtoG}).


\vspace{-10pt}
\paragraph{\bf Organization.} In Section~\ref{sec:prelim} we give all necessary definitions and state some basic facts that we need. In Subsection~\ref{sec:construction} we describe the construction of {\em the ``lifted language''}, taken from \cite{kolmogorov15:mrt}, and we  introduce a novel framework of {\em CSPs with input prototype}.
In Section~\ref{sec:relaxed} we introduce our main construction of the template $\bGamma^\mathcal{B}$ and give examples of this construction. 
In Section~\ref{sec:Main} we prove the main result of the paper, that is an algorithm for $\textsc{CSP}(\bGamma)$, given a homomorphism from an instance to $\bGamma^\mathcal{B}$ (Theorem~\ref{Btract}).
Section~\ref{sec:relation} is dedicated to a Karp reduction of $\textsc{CSP}(\bGamma^\mathcal{B})$ to a non-uniform CSP that has basically two types of relations --- those that are in $\Gamma$ and a so called trace of $\mathcal{B}$ (Theorems~\ref{simple-red} and~\ref{strong-red11}). As a consequence we obtain that $\textsc{CSP}(\bGamma^\mathcal{B})$ can be reduced to $\textsc{CSP}(\bGamma)$ if $\mathcal{B}$ is preserved under all polymorphisms of $\bGamma$. 

In Section~\ref{nice-examples} of Appendix one can find 7 examples demonstrating a usefullness of the induced problem (besides Green and Cohen's case). These examples show that for an appropriately chosen $\mathcal{B}$, the template $\bGamma^\mathcal{B}$ is tractable by construction, for any $\bGamma$. Any proof that is omitted in the main part of the paper can be found in Appendix. 


\vspace{-10pt}
\section{Preliminaries}\label{sec:prelim}
\vspace{-5pt}
A problem is called {\em tractable} if it can be solved in polynomial time. We assume ${\rm P} \neq {\rm NP}$. 
Typically, a finite domain of CSP is denoted by $D$ and a  finite set of variables is denoted by $V$. We denote tuples in lowercase boldface such as $\mathbf{a} = (a_1, \dots, a_k)$. Also for a mapping $h \colon A \to B$ and a tuple $\mathbf{a} = (a_1, \dots, a_k)$, where $a_j \in A$ for $j = 1, \dots, k$, we will write $\mathbf{b} = (h(a_1), \dots, h(a_k))$ simply as $\mathbf{b} = h(\mathbf{a})$. 
Let $\ar(\varrho)$, $\ar(\mathbf{a})$ stand for the arity of a relation $\varrho$ and the size of a tuple $\mathbf{a}$, respectively. 
A relational structure is a finite set and a tuple of relations of finite arity defined on that set.
The set $\{1,...,k\}$ is denoted by $[k]$.
 
Let us formulate the general CSP as a homomorphism problem.

\begin{definition} Let $\mathbf{R} = (D^\mathbf{R}, r^\mathbf{R}_1, \cdots, r^\mathbf{R}_s)$ and $\mathbf{R}' = (D^{\mathbf{R}'}, r^{\mathbf{R}'}_1, \cdots, r^{\mathbf{R}'}_s)$ be relational structures with a common signature (that is $\ar(r^\mathbf{R}_i)= \ar(r^{\mathbf{R}'}_i)$  for every $i \in [s]$). A mapping $h\colon D^\mathbf{R}\to D^{\mathbf{R}'}$ is called a \emph{homomorphism} from $\mathbf{R}$ to $\mathbf{R}'$ if for every $i \in [s]$ and for any $(x_1, \dots, x_{\ar(r^{\mathbf{R}}_i)}) \in r^{\mathbf{R}}_i$ we have that $\big((h(x_1), \dots, h(x_{\ar(r^{\mathbf{R}}_i)})\big) \in r^{\mathbf{R}'}_i$. In this case, we write $\mathbf{R} \stackrel{h}{\to} \mathbf{R}'$ or just $\mathbf{R} \to \mathbf{R}'$. Also, we denote ${\rm Hom}(\mathbf{R}, \mathbf{R}') = \{h\mid \mathbf{R} \stackrel{h}{\to} \mathbf{R}'\}$.
\end{definition}

A finite relational structure $\mathbf{\Gamma} = (D, \varrho_1, \dots, \varrho_s)$ over a fixed finite domain $D$ is called a template.

\begin{definition} Let $D$ be a finite set and $\bGamma$ a template over $D$. Then the {\bf fixed template CSP} for template $\bGamma$, denoted $\CSP\bGamma$, is defined as follows:
given a relational structure $\mathbf{R}$ of the same signature as $\bGamma$, find a homorphism $h:\mathbf{R} \to \mathbf{\Gamma}$.\footnote{Throughout the paper we define the CSP as a search problem, not as a decision problem, taking into account that both formulations are equivalent~\cite{Cohen2004,bulatov05:classifying}.}
\end{definition}

For $\bGamma = (D, \varrho_1, \dots, \varrho_s)$ we  denote by $\Gamma$ (without boldface) the set of relations $\{ \varrho_1, \dots, \varrho_s \}$ (which is called {\em the constraint language}) and by $\CSP\Gamma$ we  denote $\CSP\bGamma$.

\begin{definition}
A language $\Gamma$ is said to be tractable if
$\textsc{CSP}(\Gamma)$ is tractable. Also, $\Gamma$ is said to be NP-hard if $\textsc{CSP}(\Gamma)$ is NP-hard.
\end{definition}


Any language $\Gamma$ over a domain $D$ can be associated
with a set of operations on $D$,
known as the polymorphisms of $\Gamma$ \cite{barto_et_al}, defined as follows.
\begin{definition}
\label{def:polymorphism}
An operation $g: D^m\to D$ is a \emph{polymorphism} of a relation $\varrho\subseteq D^n$ (or ``$g$ \emph{preserves} $\varrho$'', or ``$\varrho$ \emph{is preserved by} $g$'') if,
for any $m\times n$-matrix whose columns are $\overline{x}_1,\ldots,\overline{x}_n$ and whose rows are all in $\varrho$,
we have  $\big(g(\overline{x}_1),\ldots,g(\overline{x}_n)\big)\in \varrho$.
For any constraint language $\Gamma$ over a set $D$,
we denote by $\pol(\Gamma)$ the set of all operations on $D$ which are polymorphisms of every
$\varrho \in \Gamma$.
\end{definition}
%

%
\ifTR 
Let us denote the set of polymorphisms of $\Gamma$ by $\pol(\Gamma)$. 
Jeavons~\cite{Jeavons:1998} showed that the complexity of $\CSP\Gamma$ is fully determined by $\pol(\Gamma)$, which was the first step in developing the so-called algebraic approach to fixed template CSP. We will also use the notation $\pol(\bGamma)$, meaning $\pol(\Gamma)$.
\else 
\fi
\vspace{-10pt}
\subsection{Multi-sorted CSPs}
\vspace{-5pt}
For any finite collection of finite domains $\mathcal{D} = \big\{D_i|i\in I\big\}$, and any list of
indices $(i_1, i_2, ...,i_m) \in I^m$, a subset $\varrho$ of $D_{i_1} \times D_{i_2} \times \cdots \times D_{i_m}$, together with the list $(i_1, i_2, ...,i_m)$, is called a multi-sorted relation over $\mathcal{D}$ with arity $m$ and signature $(i_1, i_2, ...,i_m)$. For any such relation $\varrho$, the signature of $\varrho$ is
denoted $\sigma(\varrho)$.

\begin{definition}\label{multi-sorted}
Let $\Gamma$ be a set of multi-sorted relations over a collection of sets
$\mathcal{D} = \big\{D_i|i\in I\big\}$. The multi-sorted constraint satisfaction problem over $\Gamma$,
denoted $\textsc{MCSP}(\Gamma)$, is defined to be the search problem with:

\paragraph{\bf Instance:} A triple $(V ; \delta; \mathcal{C})$ where
\begin{itemize}
\item $V$ is a set of variables;
\item $\delta$ is a mapping from $V$ to $I$, called the domain function;
\item $\mathcal{C}$ is a set of constraints, where each constraint $C\in\mathcal{C}$ is a pair $(s,\varrho)$,
such that
\begin{itemize}
\item $s = (v_1, . . . , v_{m_C})$ is a tuple of variables of length $m_C$, called the
constraint scope;
\item $\varrho$ is an element of $\Gamma$ with the arity $m_C$ and signature $(\delta(v_1), . . . , \delta(v_{m_C}))$
called the constraint relation.
\end{itemize}
\end{itemize}

\paragraph{\bf Question:} Find a solution (or, indicate its nonexistence), i.e., a function $\phi$, from $V$ to $\cup_{i\in I}D_i$
such that, for each variable $v\in V$ , $\phi(v) \in D_{\delta(v)}$, and for each constraint
$(s, \varrho) \in \mathcal{C}$, with $s = (v_1, . . . , v_{m})$, the tuple $(\phi(v_1), ..., \phi(v_m))$ belongs to $\varrho$?
\end{definition}

By construction a fixed template CSP, given in the form of a homomorphism problem, can be formulated as a multi-sorted CSP over a collection of domains $\mathcal{D} = \big\{D\big\}$. The problem of finding a homomorphism $h: \mathbf{R}\to \bGamma$ where $\mathbf{R} = (V,r_1, \dots, r_s)$ and $\bGamma = (D, \varrho_1, \dots, \varrho_s)$, is equivalent to the following set of constraints:
\begin{equation}\label{vv-recored}
\big\{
(\bv, \varrho_i) | i\in [s], \bv\in r_i
\big\}.
\end{equation}

\begin{definition}
A set of multi-sorted relations, $\Gamma$ is said to be tractable if 
$\textsc{MCSP}(\Gamma)$ is tractable.
\end{definition}

\begin{definition}
Let $\mathcal{D}$ be a collection of sets. An $m$-ary multi-sorted operation
$t$ on $\mathcal{D}$ is defined by a collection of interpretations $\big\{t^D| D\in\mathcal{D}\big\}$, where each $t^D$ is an $m$-ary operation on the corresponding set $D$. A multi-sorted operation
$t$ on $\mathcal{D}$ is said to be a polymorphism of an $n$-ary multi-sorted relation $\varrho$ over $\mathcal{D}$ with
signature $(\delta(1), . . . , \delta(n))$ if, for any $m\times n$-matrix $\left[\overline{x}_1,\ldots,\overline{x}_n\right]$ whose rows are all in $\varrho$, we have
\begin{equation}
\big(t^{D_{\delta(1)}}\big( \overline{x}_1 \big), \ldots, t^{D_{\delta(n)}}\big(\overline{x}_n\big)\big)
\in \varrho.
\end{equation}
\end{definition}
For the set of multi-sorted relations $\Gamma$, $\mpol(\Gamma)$ denotes the set of all multi-sorted operations that are polymorphisms
of each relation in $\Gamma$.

\vspace{-10pt}
\subsection{The lifted language}\label{sec:construction}
\vspace{-5pt}
Let $\bGamma=(D,\varrho_1,\ldots,\varrho_s)$ be a template and $\bR=(V,r_1,\ldots,r_s)$ be a relational structure given as an input to $\CSP\bGamma$. 
The problem of finding a homomorphism $h: \bR\to \bGamma$ can be reformulated as an instance of the multi-sorted CSP in many different ways. We choose the most straightforward one as it gives an insight into the construction of the lifted language. We introduce for every variable $v\in V$ its  unique domain $D_v = \{(v,a)|a\in D\}$. Thus, we get a collection of domains $\mathcal{D} = \{D_v| v\in V\}$.

For tuples $\mathbf{a} = (a_1, \dots, a_p) \in D^p$ and $\bv=(v_1,\ldots,v_p)\in V^p$ denote $d(\mathbf{v}, \mathbf{a}) = ((v_1, a_1), ...,$ $(v_p, a_p))$.
Now for a relation $\varrho \subseteq D^p$ and $\mathbf{v}=(v_1,\ldots,v_p)\in V^{p}$ we will define a multi-sorted relation $\varrho(\mathbf{v})$ over $\mathcal{D}$ with a signature $(v_1,\ldots,v_p)$ by
\begin{equation}\label{pred-func}
\varrho(\bv) =\big\{ d(\mathbf{v}, \mathbf{y}) | \mathbf{y}\in \varrho\big\}.
\end{equation}

The set of constraints
$\{\big(\bv, \varrho_i(\bv)\big) \::\: i\in[s], \bv \in r_i \} \cup \{ (v, D_v) \::\: v \in V\} 
$
defines an instance of the multi-sorted CSP whose solutions are in one-to-one correspondence with homomorphisms from $\bR$ to $\bGamma$. The correspondence between $h: V\to D$ and $h': V\to \cup \mathcal{D}$ is established by the rule $h'(v) = (v, h(v))$.

Finally, we construct the language $\Gamma_\bR$ (which is called the lifted language) that consists of multi-sorted relations over $\mathcal{D}$
\begin{equation}
\Gamma_\bR = \{\varrho_i(\bv) \::\: i\in[s], \bv \in r_i \} \cup \{ D_v \::\: v \in V\}.
\end{equation}
Note that, if all relations in $\bR$ are nonempty, the lifted language $\Gamma_\bR$ contains all information about a pair $\bR, \bGamma$. After ordering its relations we get the template $\bGamma_\bR$. A more general version of this language (formulated in terms of cost functions) first appeared in the context of {\em the hybrid CSPs} which is an extension of the fixed-template CSP framework (see Section~5.1 of~\cite{kolmogorov15:mrt}).

Note that this language defines $\textsc{MCSP}(\bGamma_\bR)$, in which every variable is paired with its domain as in Definition~\ref{multi-sorted}. Sometimes the lifted language will be treated as a set of relations over a common domain $\cup_{v} D_v = V\times D$ (i.e. not multi-sorted), e.g. as in Theorem~\ref{strong-red}. Since all domains $D_v, v\in V$ are disjoint, $\textsc{MCSP}(\bGamma_\bR)$ and $\textsc{CSP}(\bGamma_\bR)$ are Karp reducible to each other.

The following lemma plays a key role in our paper. It shows that  $\textsc{MCSP} (\bGamma_\bR)$ is equivalent to another problem formulation called {\em the CSP with an input prototype} (see Section 6 of~\cite{Takhanov15}). 

\begin{definition}
For a given template $\bGamma$ and a relational structure ${\bf P}$, {\bf the CSP with an input prototype} ${\bf P}$ is a problem, denoted $\textsc{CSP}^+_{{\bf P}}(\bGamma)$, for which: a) an instance is a pair $(\bR, \chi)$ where $\bR$ is a relational structure and $\chi:\bR\to {\bf P}$ is a homomorphism; b) the goal is to find a homomorphism $h: \bR\to\bGamma$. 
\end{definition}
E.g., when ${\bf P} = ([4], \ne), \bGamma = ([3], \ne)$, then $\textsc{CSP}^+_{{\bf P}}(\bGamma)$ is a problem of finding a 3-coloring of a graph whose 4-coloring is given as part of input.
\begin{lemma}\label{prototype} $\textsc{MCSP}(\bGamma_{\bf P})$ and $\textsc{CSP}^+_{{\bf P}}(\bGamma)$ are Karp reducible to each other in linear time.
\end{lemma}
\ifTR
\begin{proof} {\bf Karp reduction of $\textsc{MCSP}(\bGamma_{{\bf P}})$ to $\textsc{CSP}^+_{{\bf P}}(\bGamma)$.}
Let $\bGamma = \big(D, \varrho_1,..., \varrho_s\big)$ and ${\bf P}=(V,r_1,\ldots,r_s)$ be given. 
Let $\mathcal{I}$ be an instance of $\textsc{MCSP}(\bGamma_{{\bf P}})$ that consists of a set of variables $W$ and a set of constraints $C$. For any $\varrho_i(\bv)\in \Gamma_{{\bf P}}$ let us denote $f_i^\bv = \{\bv' | (\bv',\varrho_i(\bv))\in C\}$. Thus,
$C = \big\{(\bv',\varrho_i(\bv))| i\in[s], \bv \in r_i,\bv'\in f_i^\bv\big\}$.

Also, we are given an assignment $\delta: W\rightarrow V$, that assigns each variable $v\in W$ its domain $D_{\delta(v)}$. Denote ${\bR} = (W, f_1,...,f_s)$, where $f_i = \cup_{\bv\in r_i}f_i^\bv$. 


According to definition~\ref{multi-sorted}, $\varrho_i(\bv)$ is a relation with signature $\delta(\bv')$, $\bv'\in f_i^\bv$.
Therefore, for any $\bv'\in f_i^\bv$ its component-wise image $\delta(\bv')$ is exactly the tuple $\bv$. Since $\bv\in r_i$, we conclude ${\bR}\mathop\rightarrow\limits^{\delta} {\bf P}$. 

For $h: W\to D$, let us define $h^\delta: W \to V\times D$ by $h^\delta(v) = (\delta(v), h(v))$.
Vice versa,  to every assignment $g: W \to V\times D$ we will associate an assignment $g^f(x) = F(g(x))$ where $F$ is a ``forgetting'' function, i.e. $F((v,a)) = a$.  For any assignment $g: W \to V\times D$ that satisfies $g(v)\in D_{\delta(v)}$, by construction, we have $(g^f)^\delta=g$. 
By construction, $g$ is a solution of our multi-sorted CSP if and only if ${\bR}\mathop\rightarrow\limits^{g^f}\bGamma$.
The latter is an instance of $\textsc{CSP}^+_{{\bf P}}(\bGamma)$ with an input structure ${\bR}$ and a homomorphism $\delta:{\bR}\to {\bf P}$ given, and any solution $h$ of it corresponds to a solution $h^\delta$ of $\mathcal{I}$. By construction, all computations are polynomial-time, so we proved that $\textsc{MCSP}(\bGamma_{{\bf P}})$ can be polynomially reduced to $\textsc{CSP}^+_{{\bf P}}(\bGamma)$.
\vspace{-7pt}
\paragraph{\bf Karp reduction of $\textsc{CSP}^+_{{\bf P}}(\bGamma)$ to $\textsc{MCSP}(\bGamma_{{\bf P}})$.} Let $\bGamma = \big(D, \varrho_1,..., \varrho_s\big)$, ${\bf P}=(V,r_1,\ldots,r_s)$. Suppose we are given an instance of $\textsc{CSP}^+_{\bf P}(\bGamma)$ with an input structure $\bR = (W, f_1,...,f_s)$ and a homomorphism $\delta:\bR\to {\bf P}$, i.e. our goal is to satisfy the set of constraints
$\big\{(\bv, \varrho_i) | i\in [s], \bv\in f_i\big\}$.
Let us construct an instance of $\textsc{MCSP}(\bGamma_{{\bf P}})$:
$$\big\{(\bv, \varrho_i(\delta(\bv))) |i\in [s], \bv\in f_i\big\}, \{(v, D_{\delta(v)})| v\in W\}.$$
It is straightforward to check that if $g$ is a solution for this instance then $h = g^f$ is a solution for $\textsc{CSP}^+_{{\bf P}}(\bGamma)$ and visa versa. It remains to note that $\big\{\varrho_i(\delta(\bv)) |i\in [s], \bv\in f_i\big\}\subseteq \Gamma_{{\bf P}}$.
By construction both reductions take linear time on the size of input. \qed
\end{proof}
\else
\fi

\ifCOM
\section{Algebras induced on the solution set}
\subsection{Motivation}
In this section we will describe the main objects of our study.
For any input $\bR=(V,r_1,\ldots,r_s)$ of $\CSP\bGamma$ let us denote 
$$\textsc{Hom}(\bR,\bGamma)=\{h| h:\bR\to \bGamma\}$$ 

As we previously mentioned, the complexity of $\CSP\bGamma$ is fully determined by $\pol(\bGamma)$. Moreover, it is a well-known fact that any polymorphism $p\in \pol(\bGamma)$ also preserves the set $\textsc{Hom}(\bR,\bGamma)$. In other words, if $p$ is $n$-ary, then for any $h_1,..., h_n\in \textsc{Hom}(\bR,\bGamma)$, $h(v) = p(h_1(v),..., h_n(v))$  also is in $\textsc{Hom}(\bR,\bGamma)$.
Thus, $p$ not only gives us information on the tractability of $\CSP\bGamma$, but also defines an operation on $\textsc{Hom}(\bR,\bGamma)$. 

This way of defining an algebraic structure on $\textsc{Hom}(\bR,\bGamma)$ can be directly generalized. Let us, instead of starting from the polymorphism $p\in\pol(\bGamma)$ (that does not depend on the input $\bR$), consider more general multi-sorted polymorphisms of the lifted language $\Gamma_\bR$. Recall that a multi-sorted polymorphism $m\in \mpol(\Gamma_\bR)$ should have an interpretation $m^{D_v}$ on every domain $D_v, v\in V$ and this adds some freedom to the definition of $m$. Thus, $\mpol(\Gamma_\bR)$ is a substantially richer object than $\pol(\bGamma)$, and its study can give us more information about the structure of $\textsc{Hom}(\bR,\bGamma)$.

\begin{example}\label{cohen-horn} If $D=[d]$ and $\Gamma$ is preserved by $\max(x,y)$, then $\Gamma$ is tractable~\cite{jeavons98:algebraic}. 
A special case of CSP with such a constraint language is equivalent to the satisfiability problem with Horn (anti-Horn) clauses. A generalization of this class was proposed in~\cite{GREEN20081094}. Suppose that $\bR$ is an instance of $\CSP\bGamma$ and also we are given a collection of domain permutations $\pi_v: D\to D, v\in V$. Now let us introduce a copy of the set $V$, denoted $V' = \{v'| v\in V\}$, and add new variables $v'\in V'$ to the CSP instance, together with new binary constraints that require an assignment of a pair $(v,v')$, where $v'$ is a copy of $v$, to be in $\{(a,b)|b=\pi_v(a)\}$  (all old constraints are unchanged). There is one-to-one correspondence between solutions of old CSP and the new one (defined on a set of variables $V\cup V'$). Further, we can discard old variables, i.e. $V$, and recalculate initial constraints on $V'$ (in other words, we project a set of solutions onto new variables from $V'$). Thus, if $(v_1, \dots, v_p)\in r_i$, now assignments of $(v'_1, \dots, v'_p)$ are constrained to be in 
$\varrho'_i = \{(\pi_{v_1}(x_1), \cdots, \pi_{v_p}(x_p)) | (x_1, \cdots, x_p)\in \varrho_i\}$. This operation of discarding $V$  defines a CSP instance $\mathcal{I}$ with a set of variables $V'$. 
If, after recalculation of constraints, relations of $\mathcal{I}$ become max-closed (i.e. preserved by $\max(x,y)$), then one can solve $\mathcal{I}$, find its solution $h': V'\to D$, and recover a solution for the initial one by $h(v) = \pi_v^{-1}(h'(v'))$.

Suppose that domain permutations $\pi_v: D\to D, v\in V$ satisfy that property, i.e. all recalculated relations are max-closed. Then, it is straightforward to show that the collection $\pi_v: D\to D, v\in V$ induces a multi-sorted polymorphism $m$ of $\Gamma_\bR$ given by the following rule: $$m^{D_v}\big((v,a), (v, b)\big) = \big(v, {\max}_{\pi_v}(x,y)\big)$$ where $\max_\pi(x,y) =  \pi^{-1}(\max(\pi(a),\pi(b)))$. Conversely, any polymorphism $m\in\mpol(\Gamma_\bR)$ for which  $m^{D_v}((v,a), (v, b)) = (v, {\max}_{\pi'_v} (x,y))$, for appropriate $\pi'_v$, defines a collection of permutations for which all recalculated relations are max-closed. Thus, the idea of the domain permutation reduction to max-closed languages can be understood as a problem of finding a polymorphism $m\in\mpol(\Gamma_\bR)$ of a certain kind.
\end{example}



Thus, we will be especially interested in $n$-ary polymorphisms $m\in\mpol(\Gamma_\bR)$ for which $m^{D_v} \big((v, a_1), ..., (v,a_{n})\big) = \big(v, o(a_1, ..., a_{n})\big)$ where $o$ will be taken from a set of operations specified in advance. Below we show how this idea can be implemented in the most general form, i.e. for algebras.
\else
\fi

\vspace{-16pt}
\section{The construction}\label{sec:relaxed}
\vspace{-5pt}
Suppose that we are given a list $o_1,...,o_k$ of symbols with prescribed arities $n_1,...,n_k$. This list is called the signature and denoted $\sigma$. 
An {\em algebra} with a signature $\sigma$ is a tuple $\mathbb{A} = \big(D^\mathbb{A}, o^{\mathbb{A}}_1,o^\mathbb{A}_2,...,o^\mathbb{A}_k\big)$, where $D^\mathbb{A}$ denotes a finite domain of the algebra and $o^\mathbb{A}_i: \big(D^\mathbb{A}\big)^{n_i}\rightarrow D^\mathbb{A}, i\in [k]$ denote its basic operations. 
Let us denote by $\mathcal{A}_D^\sigma$ the set of algebras with one fixed signature $\sigma$ and over a single fixed domain $D$.
Suppose we are given a collection $\mathcal{B}\subseteq \mathcal{A}_{D}^\sigma$ and a relational structure $\bGamma = \big(D, \varrho_1, ..., \varrho_s\big)$ where $\varrho_i$ is a relation 
over $D$. 

\begin{definition}\label{def-main1}
Let $\varrho$ be an $m$-ary relation over $D$. Let us define
$$
\varrho^{\mathcal{B}} = \{(\mathbb{A}_1,...,\mathbb{A}_m)\in \mathcal{B}\times \cdots \times \mathcal{B} \mid \varrho{\rm \,\,is\,\,a\,\,subalgebra\,\,of\,\,}\mathbb{A}_1\times \cdots\times\mathbb{A}_m\}.
$$
In other words, we define the relation $\varrho^{\mathcal{B}}$ as a subset of $\mathcal{B}^m$ that consists of tuples $(\mathbb{A}_1,...,\mathbb{A}_m)\in \mathcal{B}^m$ such that for any $i\in [k]$, $\big(o^{\mathbb{A}_1}_i,o^{\mathbb{A}_2}_i,...,o^{\mathbb{A}_m}_i\big)$ is a component-wise polymorphism of $\varrho$. \footnote{$\varrho^\mathcal{B}$ can be empty.} \,\, The last condition means that for any matrix $\left[\overline{x}_1,\ldots,\overline{x}_m\right]\in D^{n_i\times m}$ whose rows are all in $\varrho$, we have $\big(o^{\mathbb{A}_1}_i(\overline{x}_1)$, $o^{\mathbb{A}_2}_i(\overline{x}_2)$, ..., $o^{\mathbb{A}_m}_i(\overline{x}_m)\big)\in \varrho$. 
\end{definition}

\begin{definition}\label{def-main2}
Given $\Gamma$ and $\mathcal{B}$, we define $\Gamma^{\mathcal{B}} = \big\{\varrho^\mathcal{B} | \varrho\in \Gamma\big\}$. Analogously, if $\bGamma = \big(D, \varrho_1, ..., \varrho_s\big)$ where $\varrho_i$ is a relation 
over $D$, then we define $\bGamma^\mathcal{B} = \big(\mathcal{B}, \varrho^\mathcal{B}_1, \cdots, \varrho^\mathcal{B}_s\big)$. 
\end{definition}

Now, given an instance $\bR$ of $\textsc{CSP}(\bGamma)$ we can consider $\bR$ as an instance of $\textsc{CSP}(\bGamma^{\mathcal{B}})$. Let us decode Definitions~\ref{def-main1} and~\ref{def-main2}. Any $h\in \textsc{Hom}(\bR,\bGamma^{\mathcal{B}})$ assigns to every variable $v\in V$ an algebra $h(v)\in \mathcal{B}$. For $j\in [s]$, $\bv\in r_j$ our assignment satisfies $h(\bv)\in \varrho_j^{\mathcal{B}}$, i.e. if $\bv = (v_1, ..., v_p)$, then $(o^{h(v_1)}_{i}, ..., o^{h(v_p)}_{i})$ component-wise preserves $\varrho_j$. Suppose now that for any $v\in V$ we create a unique copy of the domain $D$, i.e. $D_v = \{(v,a)|a\in D\}$, and define $m_i^{D_v}$ as an interpretation of $o^{h(v)}_{i}$ on this copy $D_v$, i.e.
$$
m_i^{D_v}\big((v, a_1), ..., (v,a_{n_i})\big) = \big(v, o^{h(v)}_{i}(a_1, ..., a_{n_i})\big).
$$
Since for $j\in [s]$ and $\bv = (v_1, ..., v_p)\in r_j$, $(o^{h(v_1)}_{i}, ..., o^{h(v_p)}_{i})$ component-wise preserves $\varrho_j$, operation $m_i$ preserves the multi-sorted relation $\varrho_j(v_1, ..., v_p)$ (see equation~\eqref{pred-func}). In other words, $m_{i}$ is a multi-sorted polymorphism of the lifted language $\Gamma_\bR$. Thus, every assignment 
$h\in \textsc{Hom}(\bR,\bGamma^{\mathcal{B}})$ induces a system of multi-sorted polymorphisms $m_1, ..., m_k\in \mpol(\Gamma_\bR)$. 

For special cases of $\mathcal{B}$, the structure of the template $\bGamma^{\mathcal{B}}$ has been studied in~\cite{takhanov:LIPIcs:2017:8247}. 
\vspace{-10pt}
\subsection{Example: Binary and conservative operations}\label{ExCohen} 
This example is a direct generalization of Proposition 36 from~\cite{GREEN20081094}.
Let us define $\mathcal{B}$ as the set of all algebras with a commutative and conservative binary operation over $D$, i.e.
$$
\mathcal{B} = \{(D, b)| b(x,y)\in \{x,y\}, b(x,y)=b(y,x)\}.
$$
It is a well-known fact that any commutative and conservative binary operation corresponds to a tournament on $D$, i.e. to a complete directed graph with a set of vertices $D$ in which antiparallel arcs are not allowed (an identity $b(x,y)=y$ for distinct $x,y\in D$ corresponds to an arc from $x$ to $y$ in the tournament). Thus, $|\mathcal{B}|=2^{|D| \choose 2}$.
We define a ternary operation $m$ on the set $\mathcal{B}$ that acts as follows:
$$
m((D, b_1),(D, b_2),(D, b_3)) = (D, b) \Leftrightarrow b(x,y)=b_1(b_2(x,y),b_3(x,y)).
$$
By construction $m$ outputs an element from $\mathcal{B}$, i.e. $m: \mathcal{B}^3\rightarrow \mathcal{B}$. It is straightforward to check that the following identities hold for conservative and commutative operations:
$$
b_1(b_1(x,y),b_2(x,y))=b_1(b_2(x,y),b_1(x,y)) = b_2(b_1(x,y),b_1(x,y)) =b_1(x,y)
$$
or
$m({\mathbb A},{\mathbb A},{\mathbb B})=m({\mathbb A},{\mathbb B},{\mathbb A})=m({\mathbb B},{\mathbb A},{\mathbb A})={\mathbb A},\,\,\,\,  \forall{\mathbb A},{\mathbb B}\in \mathcal{B}$.

Thus, $m$ is a majority operation.

Now let us prove that for any $\varrho\subseteq D^k$, $m$ is a polymorphism of $\varrho^{\mathcal{B}}$. Indeed, suppose that $({\mathbb A}_1, \cdots, {\mathbb A}_k)\in \varrho^{\mathcal{B}}$, $({\mathbb B}_1, \cdots, {\mathbb B}_k)\in \varrho^{\mathcal{B}}$ and $({\mathbb C}_1, \cdots, {\mathbb C}_k)\in \varrho^{\mathcal{B}}$. We denote ${\mathbb A}_i = (D, a_i), {\mathbb B}_i = (D, b_i), {\mathbb C}_i = (D, c_i)$. By definition of $\varrho^{\mathcal{B}}$, we know that 
$$
\overline{x} = \begin{bmatrix}
x_1 \\
\cdots \\
x_k
\end{bmatrix}, \overline{y} = \begin{bmatrix}
y_1 \\
\cdots \\
y_k
\end{bmatrix}\in \varrho \Rightarrow \begin{bmatrix}
b_1(x_1, y_1) \\
\cdots \\
b_k(x_k, y_k)
\end{bmatrix}, \begin{bmatrix}
c_1(x_1, y_1) \\
\cdots \\
c_k(x_k, y_k)
\end{bmatrix} \in \varrho.
$$
Therefore, component-wise application of $(a_1,\cdots, a_k)$ to the last two tuples also will result in a tuple from $\varrho$:
\vspace{-10pt}
$$
 \begin{bmatrix}
a_1(b_1(x_1, y_1), c_1(x_1, y_1)) \\
\cdots \\
a_k(b_k(x_k, y_k), c_k(x_k, y_k))
\end{bmatrix} \in \varrho.
$$
The latter implies that $\big(m({\mathbb A}_1,{\mathbb B}_1,{\mathbb C}_1), \cdots, m({\mathbb A}_k,{\mathbb B}_k,{\mathbb C}_k)\big)\in \varrho^{\mathcal{B}}$.

Since, $m$ is a majority polymorphism of any $\varrho^{\mathcal{B}}$, the problem $\textsc{CSP}(\bGamma^{\mathcal{B}})$ is tractable for any $\bGamma$~\cite{jeavons98:algebraic}. Moreover, $\textsc{CSP}(\bGamma^{\mathcal{B}})$ can be solved by a local consistency checking algorithm. Note that if we define $\mathcal{B}^o \subset \mathcal{B}$ as the set of all tournament pairs that correspond to total orders on $D$, then $\textsc{CSP}(\bGamma^{\mathcal{B}^o})$ is NP-hard in general (see Proposition 38 from~\cite{GREEN20081094}). 

This example shows that $\mathcal{B}$ can be such that $\textsc{CSP}(\bGamma^{\mathcal{B}})$ is tractable for any (possibly NP-hard) $\bGamma$. Other examples of this kind can be found in  Section~\ref{nice-examples} of Appendix.

\ifCOM
\begin{theorem} Let $h$ be a homomorphism from $\bR$ to $\bGamma^{\mathcal{B}}$. Also, let ${\mathbb A}$ be an algebra with signature $\sigma$ over the domain (the set of mappings) $D^V$, whose basic operations are defined by:
$$
o^{{\mathbb A}}_i (h_1, ..., h_{n_i})(v) = o^{h(v)}_i (h_1(v), ..., h_{n_i}(v))
$$
for any $i\in [k]$ and $h_1, ..., h_{n_i}\in D^V$. Then, $\textsc{Hom}(\bR,\bGamma)$ is a subalgebra of ${\mathbb A}$, i.e. $\textsc{Hom}(\bR,\bGamma)$ is preserved by all basic operations of ${\mathbb A}$.
\end{theorem} 
\ifTR
\begin{proof}
We only need to check that for any $i\in [k]$ and $h_1, ..., h_{n_i}\in \textsc{Hom}(\bR,\bGamma)$ we have $o^{{\mathbb A}}_i (h_1, ..., h_{n_i})\in \textsc{Hom}(\bR,\bGamma)$.

If $h_l\in \textsc{Hom}(\bR,\bGamma)$, $l\in [n_i]$, then for $j\in [s]$, $\bv\in r_j$, we have that $h_l(\bv)\in \varrho_j$. But since $h\in \textsc{Hom}(\bR,\bGamma^{\mathcal{B}})$, we also have that $h(\bv)\in \varrho_j^{\mathcal{B}}$.
If $\bv = (v_1, ..., v_p)$, then $(o^{h(v_1)}_{i}, ..., o^{h(v_p)}_{i})$ component-wise preserves $\varrho_j$. From the latter we obtain that $$\big(o^{h(v_1)}_i (h_1(v_1), ..., h_{n_i}(v_1)), ..., o^{h(v_p)}_i (h_1(v_p), ..., h_{n_i}(v_p))\big)\in \varrho_j$$
Thus, $o^{{\mathbb A}}_i (h_1, ..., h_{n_i})(\bv)\in \varrho_j$ and $o^{{\mathbb A}}_i (h_1, ..., h_{n_i})\in \textsc{Hom}(\bR,\bGamma)$.
\end{proof}
\else
\fi
When we restrict domain and all operations of ${\mathbb A}$ to the set $\textsc{Hom}(\bR,\bGamma)$ we obtain a new algebra ${\mathbb A}|_{\textsc{Hom}(\bR,\bGamma)}$.
Thus, searching for an assignment $h\in \textsc{Hom}(\bR,\bGamma^{\mathcal{B}})$ is also equivalent to searching for an algebraic structure on $\textsc{Hom}(\bR,\bGamma)$ of a special kind.
Our paper is dedicated to the study of the relationship between the problems $\textsc{CSP}(\bGamma)$ and $\textsc{CSP}(\bGamma^{\mathcal{B}})$.

The examples that we list in Appendix show that the tractability of $\textsc{CSP}(\bGamma^{\mathcal{B}})$ is an interesting issue in its own right. In the next section we will show that there is a relationship between complexities of $\textsc{CSP}(\bGamma)$ and $\textsc{CSP}(\bGamma^{\mathcal{B}})$.  Then we will return to our examples in Section~\ref{sec:Main}.
\else
\fi
\ifCOM
\section{Overview of results}\label{overview}
Before describing technicalities let us briefly outline the main results established in the paper. 
In Section~\ref{sec:relation} we are interested in the complexity of $\textsc{CSP}(\bGamma^{\mathcal{B}})$ and how it is affected by the structure of $\bGamma$ and $\mathcal{B}$. 

A set of assignments $\{h:V\to D\}$ of a set of variables $V = \{(o_i,a)| i\in [k], a\in D^{n_i}\}$ is in one-to-one correspondence with a set of algebras with a signature $\sigma=(n_1, \cdots, n_k)$. Thus, any primitive positive formula $\Psi(V) = \exists \, A \Phi(V,A)$ over $\Gamma$ with a set of variables $V\cup A$ (where $A$ is a set of additional variables) defines some collection of algebras $\mathcal{B}[\Psi]\subseteq \mathcal{A}_D^\sigma$.
We prove theorems~\ref{simple-red} and~\ref{strong-red} based on a natural reduction of $\textsc{CSP}(\bGamma^{\mathcal{B}})$ to $\textsc{CSP}(\bGamma)$ and deduce the corollary~\ref{BtoG} that is equivalent to the following proposition:

{\bf Proposition 1.} For any primitive positive formula $\Psi(V)$ over $\Gamma$, $\textsc{CSP}(\bGamma^{\mathcal{B}[\Psi]})$ is Karp reducible to $\textsc{CSP}(\bGamma)$.

In fact, many interesting algebra collections are pp-definable over $\{=\}$ only, as in the following example.
\begin{remark} \label{simpleEq}  Let $T$ be a set of identities such that every identity in $T$ is of the form $o_i(x_1,...,x_{n_i})=o_j(x_{n_i+1},...,x_{n_i+n_j})$, where $i,j \in \big\{1,...,k\big\}$, $x_1,...,x_{n_i+n_j}$ are variables with repetitions allowed (such identities are sometimes called Maltsev conditions of height 1). Suppose $\mathcal{B}$ is defined to be the set of algebras from $\mathcal{A}_{D}^\sigma$ whose operations satisfy the identities from $T$. By construction, the pp-formula 
$$\Psi_T(V) = \bigwedge\limits_{\begin{matrix}
\scriptscriptstyle (o_i(x_{1:n_i})=o_j(x_{n_i+1:n_i+n_j}))\in T \\
\scriptscriptstyle h:\{x_1,...,x_{n_i+n_j}\}\to D
\end{matrix}} [(o_i,h(x_1),...,h(x_{n_i}))=(o_j,h(x_{n_i+1}),...,h(x_{n_i+n_j}))],$$ 
defines $\mathcal{B} = \mathcal{B}[\Psi_T]$.

Moreover, if $\Gamma$ contains all singletons, i.e. $\{a\}\in \Gamma$ for any $a\in D$, then adding atoms
$$
\bigwedge\limits_{h:\{x_1,...,x_{n_i},x_j\}\to D}[(o_i, h(x_1),...,h(x_{n_i})) = h(x_j)]
$$
to $\Psi_T (V)$ restricts an algebra $o^{\mathcal{B}[\Psi_T]}_i$ to satisfy $o_i(x_1,...,x_{n_i}) = x_j$. Thus, well-known equational classes corresponding to tractable cases of fixed-template CSP (e.g. $T = \{o_1(x,x,y)=o_1(y,x,x)=y\}$ for Maltsev algebras, $T = \{o_1(x,x,y)=o_1(y,x,x)=o_1(x,y,x)=x\}$ for near-unanimity algebras etc.) define algebra collections $\mathcal{B}[\Psi_T]$ pp-definable over $\{=\}\cup \{\{a\}| a\in D\}$. Thus, $\textsc{CSP}(\bGamma^{\mathcal{B}})$ is Karp reducible to $\textsc{CSP}(\bGamma)$ for all such collections. For example, let $\Gamma$ be any tractable constraint language over $D$ and $\mathcal{B} = \{(D, \phi)| \phi: D^2\to D{\rm \,\,is\,\,commutative}\}$. Then, $\textsc{CSP}(\bGamma^{\mathcal{B}})$ is tractable, because $\textsc{CSP}(\bGamma)$ is tractable.
\end{remark} 

The last remark shows that for well-known collections consisting of tractable algebras (Maltsev, near unanimity, binary commutative etc), $\textsc{CSP}(\bGamma^{\mathcal{B}})$ cannot be harder than $\textsc{CSP}(\bGamma)$. This makes promising the idea of reducing $\textsc{CSP}(\bGamma)$ to $\textsc{CSP}(\bGamma^{\mathcal{B}})$. The main tool in such a reduction is Theorem~\ref{Btract} which can be reformulated as:

{\bf Proposition 2.} If $\mathcal{B}\subseteq \mathcal{A}_D^\sigma$ is a collection of tractable algebras, then $\textsc{CSP}^+_{\bGamma^{\mathcal{B}}}(\bGamma)$ is tractable for any $\bGamma$.

\begin{remark}
In examples 4.1-4.6, $\textsc{CSP}(\bGamma^{\mathcal{B}})$ is tractable for any (possibly NP-hard) $\bGamma$. Suppose 
an instance of $\textsc{CSP}(\bGamma)$ is homomorphic to $\bGamma^{\mathcal{B}}$ (this homomorphism we can find efficiently). Note that in all those examples, $\mathcal{B}$ consists of tractable algebras. Thus, the latter proposition can be directly applied in the listed cases. 
The tractability of $\textsc{CSP}^+_{\bGamma^{\mathcal{B}}}(\bGamma)$ means that we can solve $\textsc{CSP}(\bGamma)$, because an instance is equipped with a homomorphism to $\bGamma^{\mathcal{B}}$. 
\end{remark}

A set of instances homomorphic to $\bGamma^{\mathcal{B}}$ can be naturally extended, using Theorem~\ref{generalization}, which simplifies to the following proposition.

{\bf Proposition 3.} Let $\mathcal{B}\subseteq \mathcal{A}_D^\sigma$ be a collection of tractable algebras. Let $p:\{\bR\}\rightarrow \{\bR\}\cup \{{\bf UNDEF}\}$ be a preprocessing of instances of $\textsc{CSP}(\bGamma)$ that takes a polynomial time on the size of an input. Then, solving $\textsc{CSP}(\bGamma)$ with an instance from 
$$\Sigma = \{\bR | \bGamma_{p(\bR)}\to \big(\bGamma^{\mathcal{B}}\big)_{p(\bR)}\},$$
can be polynomially reduced to $\textsc{CSP}(\bGamma^{\mathcal{B}})$.

If $p(\bR) = \bR$ (i.e. we do not make any preprocessing), then the proposition helps in solving a CSP if its instance is from:
$$\Sigma = \{\bR | \bGamma_{\bR}\to \big(\bGamma^{\mathcal{B}}\big)_{\bR}\}\supseteq \{\bR| \bR\to \bGamma^{\mathcal{B}}\}.$$
\else
\fi

\vspace{-10pt}
\section{The complexity of $\textsc{CSP}^+_{\bGamma^{\mathcal{B}}}(\bGamma)$}\label{sec:Main}
\vspace{-5pt}
Again, let $\Gamma$ be a set of relations over $D$ 
and $\mathcal{B}\subseteq \mathcal{A}_{D}^\sigma$. 
In the previous section we gave an example of a set of algebras $\mathcal{B}$ for which $\textsc{CSP}(\bGamma^{\mathcal{B}})$ is tractable for any $\bGamma$ (possibly NP-hard). More such examples can be found in Appendix. 
Therefore, asking what can be achieved by substituting $\bGamma$ with another template $\bGamma'$ (e.g. $\bGamma' = \bGamma^{\mathcal{B}}$ for an appropriate $\mathcal{B}$), and the consequences of such substitutions,  may be a promising research direction.

First we will study the following problem: if we managed to find a homomorphism from the input structure $\bR$ to $\bGamma^{\mathcal{B}}$, when can it help us to find $h: \bR\to \bGamma$? In subsection~\ref{sec:construction} we called this problem {\em the CSP with an input prototype} and denoted as $\textsc{CSP}^+_{\bGamma^{\mathcal{B}}}(\bGamma)$. It was prove to be equivalent to  $\textsc{CSP}(\bGamma_{\bGamma^{\mathcal{B}}})$. Thus, we will start with identifying conditions for the tractability of $\bGamma_{\bGamma^{\mathcal{B}}}$.

\vspace{-10pt}
\subsection{Conditions for the tractability of $\bGamma_{\bGamma^{\mathcal{B}}}$}
\vspace{-5pt}

From the definition of lifted languages it is clear that for any relation $P \in \Gamma_{\bGamma^{\mathcal{B}}}$ there exists a relation $\varrho\in\Gamma$ such that $P=\varrho (\mathbb{A}_1, \cdots, \mathbb{A}_p)$ where $(\mathbb{A}_1, \cdots, \mathbb{A}_p)\in \varrho^{\mathcal{B}}$. Since $\varrho^{\mathcal{B}}\subseteq {\mathcal{B}}^p$, we have 
$\mathbb{A}_1\in {\mathcal{B}}, \cdots, \mathbb{A}_p\in {\mathcal{B}}$. From the definition of $\varrho (\mathbb{A}_1, \cdots, \mathbb{A}_p)$ (see~\eqref{pred-func}) we obtain that
$$\varrho (\mathbb{A}_1, \cdots, \mathbb{A}_p)\subseteq \{(\mathbb{A}_1, x_1)| x_1\in D\}\times \cdots \times \{(\mathbb{A}_p, x_p)| x_p\in D\}.$$

For any algebra $\mathbb{A}\in \mathcal{B}$ let us denote by $\mathbb{A}^{\bf{c}}$ a copy of $\mathbb{A}$, but with all operations redefined on its new unique domain $D^{\mathbb{A}^{\bf{c}}} = \{(\mathbb{A}, x)| x\in D\}$ by
$$
o^{\mathbb{A}^{\bf{c}}}_l\big((\mathbb{A}, x_1), \cdots, (\mathbb{A}, x_{n_l})\big)
 = (\mathbb{A}, o^{\mathbb{A}}_l(x_1, \cdots, x_{n_l})).
$$
If we introduce a new collection of domains $\big\{D^{\mathbb{A}^{\bf{c}}}|\mathbb{A}\in \mathcal{B}\big\}$ parameterized by algebras from $\mathcal{B}$, then
$$\varrho (\mathbb{A}_1, \cdots, \mathbb{A}_p)\subseteq D^{\mathbb{A}^{\bf{c}}_1}\times \cdots \times D^{\mathbb{A}^{\bf{c}}_p}$$
becomes a multi-sorted relation with signature $(\mathbb{A}_1, \cdots, \mathbb{A}_p)$.
Thus, in the new notations, any relation from $ \Gamma_{\bGamma^{\mathcal{B}}}$ becomes multi-sorted over  a collection of sets $\big\{D^{\mathbb{A}^{\bf{c}}}|\mathbb{A}\in \mathcal{B}\big\}$. Also, denote $\mathcal{B}^{\bf{c}} = \big\{\mathbb{A}^{\bf{c}} | \mathbb{A}\in \mathcal{B}\big\}$.

\begin{definition}
For a collection $\mathcal{B}$, we define $\textsc{MInv}(\mathcal{B}^{\bf{c}})$ as the set of all multi-sorted
relations $\varrho$ over a collection of sets $\big\{D^{\mathbb{A}^{\bf{c}}}|\mathbb{A}\in \mathcal{B}\big\}$ such that for any $l\in [k]$, the $n_l$-ary multi-sorted operation $\big\{o^{\mathbb{A}^{\bf{c}}}_l| \mathbb{A}\in \mathcal{B}\big\}$ is a polymorphism of $\varrho$.
\end{definition}

\begin{lemma} \label{lemma:key}$\Gamma_{\bGamma^{\mathcal{B}}}$ understood as a multi-sorted language over a collection of domains $\big\{D^{\mathbb{A}^{\bf{c}}}| \mathbb{A}\in \mathcal{B}\big\}$ is a subset of $\textsc{MInv}(\mathcal{B}^{\bf{c}})$.
\end{lemma}
\begin{proof}
Let us check that $\varrho_i (\mathbb{A}_1,...,\mathbb{A}_p)\in \textsc{MInv}(\mathcal{B}^{\bf{c}})$ whenever $(\mathbb{A}_1,...,\mathbb{A}_p)\in \varrho_i^{\mathcal{B}}$. The latter premise implies that for $h\in [k]$, $\big(o^{\mathbb{A}_1}_h,...,o^{\mathbb{A}_p}_h\big)$ is a component-wise polymorphism of $\varrho_i$, and therefore, 
$\big(o^{\mathbb{A}^{{\bf c}}_1}_h,...,o^{\mathbb{A}^{{\bf c}}_p}_h\big)$ is a component-wise polymorphism of $\varrho_i (\mathbb{A}_1,...,\mathbb{A}_p)$. From the last we conclude that
the multi-sorted operation $\big\{o^{\mathbb{A}^{\bf{c}}}_h|\mathbb{A}\in \mathcal{B}\big\}$ is a polymorphism of $\varrho_i (\mathbb{A}_1,...,\mathbb{A}_p)$, or $\varrho_i (\mathbb{A}_1, \cdots,\mathbb{A}_p)\in \textsc{MInv}(\mathcal{B}^{\bf{c}})$. \qed
\end{proof}

The following definition is very natural.
\begin{definition}\label{col-tract}
A collection $\mathcal{B}$ is called {\bf tractable} if $\textsc{MInv}(\mathcal{B}^{\bf{c}})$ is a tractable constraint language.
\end{definition}
Thus, using lemma~\ref{lemma:key} and definition~\ref{col-tract} we obtain the following key result.
\begin{theorem}\label{Btract}
If a collection $\mathcal{B}$ is tractable, then $\textsc{MCSP} (\Gamma_{\bGamma^{\mathcal{B}}})$ is tractable.
\end{theorem}

\ifTR
\begin{remark}\label{equivalence-to-individual} We gave a definition of the tractable collection $\mathcal{B}$ that serves our purposes.  
It can be shown that a collection $\mathcal{B}$ is tractable if and only if every algebra $\mathbb{A}\in \mathcal{B}$ is tractable (i.e. $\textsc{Inv}(\{o^{\mathbb{A}}_l| l\in [k]\})$ is a tractable language). The fact is well-known in CSP studies, so we omit a proof of it. \end{remark}
\else
\fi
\begin{remark} In examples \ref{ExCohen}-\ref{TWNU}, given in Appendix, sets of algebras $\mathcal{B}$, $\mathcal{B}_1$, $\mathcal{B}_{\rm com}$, $\mathcal{B}_{\rm nu}$, $\mathcal{B}_{\rm M}$, $\mathcal{B}_{\rm wnu}$ are tractable. 
Then, Theorem~\ref{Btract} and Lemma~\ref{prototype} give us the tractability of $\textsc{CSP}^+_{\bGamma^{\mathcal{B}}}(\bGamma)$, $\textsc{CSP}^+_{\bGamma^{\mathcal{B}_1}}(\bGamma)$, $\textsc{CSP}^+_{\bGamma^{\mathcal{B}_{\rm com}}}(\bGamma)$, $\textsc{CSP}^+_{\bGamma^{\mathcal{B}_{\rm nu}}}(\bGamma)$, $\textsc{CSP}^+_{\bGamma^{\mathcal{B}_{\rm M}}}(\bGamma)$, $\textsc{CSP}^+_{\bGamma^{\mathcal{B}_{\rm wnu}}}(\bGamma)$. The only problem that remains now is finding a homomorphism from an input structure to $\bGamma^{\mathcal{B}}$, $\bGamma^{\mathcal{B}_1}$, $\bGamma^{\mathcal{B}_{\rm com}}$, $\bGamma^{\mathcal{B}_{\rm nu}}$, $\bGamma^{\mathcal{B}_{\rm M}}$, $\bGamma^{\mathcal{B}_{\rm wnu}}$. In the examples we discussed that these tasks are all tractable.
\end{remark}

\begin{theorem} \label{mainreduction1} If a collection $\mathcal{B}$ is tractable and $\bGamma\to \bGamma^{\mathcal{B}}$, then 
$\textsc{CSP}(\bGamma)$ is polynomial\,-\,time Turing reducible to $\textsc{CSP}(\bGamma^{\mathcal{B}})$.
\end{theorem}
\begin{proof}
For an instance $\bR$ of $\textsc{CSP}(\bGamma)$, if $\bGamma\to \bGamma^{\mathcal{B}}$, then we can replace the right template $\bGamma$ with $\bGamma^{\mathcal{B}}$ and obtain a relaxed version of the initial CSP. Suppose that we are able to solve $\textsc{CSP} (\bGamma^{\mathcal{B}})$. If the solution set of the relaxed problem is empty, then it all the more is empty for the initial one. But if we manage to find a single homomorphism from the input structure $\bR$ to $\bGamma^{\mathcal{B}}$, then the problem of finding a homomorphism $\bR\to \bGamma$ can be presented as an instance of $\textsc{CSP} _{\bGamma^{\mathcal{B}}}^+(\bGamma)$, or, by Lemma~\ref{prototype}, of $\textsc{MCSP} (\bGamma_{\bGamma^{\mathcal{B}}})$. Now, from the tractability of $\mathcal{B}$ and Theorem~\ref{Btract} we get that $\textsc{MCSP} (\bGamma_{\bGamma^{\mathcal{B}}})$ is tractable and we efficiently find a homomorphism $h:\bR\to \bGamma$. \qed
\end{proof}

Unfortunately, it is hard to satisfy the condition $\bGamma\to \bGamma^{\mathcal{B}}$ unless ${\mathcal{B}}$  contains constants. In examples \ref{tournament}-\ref{TWNU} (and in an example from Section~\ref{ExCohen}) it is not satisfied.

\vspace{-10pt}
\section{Reductions of $\textsc{CSP}(\bGamma^{\mathcal{B}})$ to $\textsc{CSP}(\bGamma)$}\label{sec:relation}
\vspace{-5pt}
Let us now find some conditions on $\mathcal{B}$ under which $\textsc{CSP}(\bGamma^{\mathcal{B}})$ is a fragment of $\textsc{CSP}(\bGamma)$. 
Again, we are given $\bGamma = \big(D, \varrho_1, ..., \varrho_s\big)$ where $\varrho_l$ is a relation 
over $D$, $l\in [s]$ and $\mathcal{B}\subseteq \mathcal{A}_{D}^\sigma$. 
In this section we will show that under very natural conditions on $\mathcal{B}$, any instance of $\textsc{CSP}(\bGamma^{\mathcal{B}})$ can be 
turned into an instance of $\textsc{CSP}(\bGamma)$. Let us introduce some natural definitions that will serve our purpose.\footnote{The notion of the trace introduced below is used in a proof of Galois theory for functional and relational clones.} Let $D = [d]$. Given $n$, let $\alpha_n(1)$, $\alpha_n(2)$, ..., $\alpha_n(d^n)$ be a lexicographic ordering of $D^n$. 

\begin{definition}\label{trace-def} The trace of $\mathcal{B}$, denoted $Tr(\mathcal{B})$, is the relation $$\Big\{\big(o_1^{{\mathbb A}}(\alpha_{n_1}(1)), \cdots, o_1^{{\mathbb A}}(\alpha_{n_1}(d^{n_1})), \cdots, o^{{\mathbb A}}_k(\alpha_{n_k}(1)), \cdots, o^{{\mathbb A}}_k(\alpha_{n_k}(d^{n_k}))\big) | {\mathbb A}\in\mathcal{B}\Big\}.$$ 
The arity of $Tr(\mathcal{B})$ is ${\textnormal {\textkappa}}(\mathcal{B}) = \sum_{s=1}^{k} d^{n_s}$.
\end{definition}

Given a number $n$ and an $m$-ary relation $\varrho$ over $D$, let us denote $\varrho^n$ the set of all tuples $(\overline{x}_1, ..., \overline{x}_m)$, where $\overline{x}_l\in D^{n}, l\in [m]$, such that all rows of the matrix $[\overline{x}_1, ..., \overline{x}_m]$ are in $\varrho$. 
Note that $\varrho^n$ is a relation over $D^n$. According to the standard terminology of universal algebra, $\varrho^n$ is a direct product $\prod_{i=1}^n\varrho$ of $m$-ary relations. 
It satisfies $|\varrho^n| = |\varrho|^n$. 
Also, let $\bGamma\smallfrown\mathcal{B} = \big(D, \varrho_1, \cdots, \varrho_s, Tr(\mathcal{B}) \big)$.

\begin{theorem}\label{simple-red} $\textsc{CSP}(\bGamma^{\mathcal{B}})$ is Karp reducible to $\textsc{CSP}(\bGamma\smallfrown\mathcal{B})$.
\end{theorem}
\ifTR
\begin{proof}
Suppose that we are given an instance of $\textsc{CSP}(\bGamma^{\mathcal{B}})$, i.e. an input structure $\bR=(V,r_1,\ldots,r_s)$. 
Our goal is to find a homomorphism $h: \bR\rightarrow \bGamma^{\mathcal{B}}$, i.e. to assign every variable $v\in V$ an algebra $h(v)\in\mathcal{B}$ in such a way that certain constraints are satisfied. Let us define an instance of $\textsc{CSP}(\bGamma\smallfrown\mathcal{B})$ with a set of variables
$$
W = \{{\mathsf a}[v,i,\alpha_{n_i}(j)]\,| v\in V, i\in [k], j\in [d^{n_i}]\}.
$$
All variables take their values in $D$. A value assigned to the variable ${\mathsf a}[v,i,\alpha_{n_i}(j)]$ corresponds to $o^{h(v)}_i(\alpha_{n_i}(j))$ for $h: \bR\rightarrow \bGamma^{\mathcal{B}}$.
For any $v\in V$, an assignment of a tuple of variables 
\begin{equation*}
\begin{split}
T(v) = \big({\mathsf a}[{v,1,\alpha_{n_1}(1)}], ..., {\mathsf a}[{v,1,\alpha_{n_1}(d^{n_1})}], 
..., {\mathsf a}[{v,k,\alpha_{n_k}(1)}], ..., {\mathsf a}[{v,k,\alpha_{n_k}(d^{n_k})}]\big)
\end{split}
\end{equation*}
is constrained to be in $Tr(\mathcal{B})$. Any such constraint models an assignment of an algebra from $\mathcal{B}$ to $v$ in the initial instance of $\textsc{CSP}(\bGamma^{\mathcal{B}})$. Indeed, an assignment of an algebra $h(v)\in \mathcal{A}_{D}^\sigma$ to $v$ corresponds to assigning the tuple $T(v) $ the value
$$\big(o^{h(v)}_1(\alpha_{n_1}(1)), ..., o^{h(v)}_1(\alpha_{n_1}(d^{n_1})), ..., o^{h(v)}_k(\alpha_{n_k}(1)), ..., o^{h(v)}_k(\alpha_{n_k}(d^{n_k}))\big).$$
The latter tuple is in $Tr(\mathcal{B})$ if and only if $h(v)\in\mathcal{B}$.

The second type of constraints are 
$$
\langle \big({\mathsf a}[{v_1,j,\overline{x}_1}], ..., {\mathsf a}[{v_p,j,\overline{x}_{p}}]\big), \varrho_l \rangle
$$
for all $l\in [s]$, $(v_1,...,v_p)\in r_l$, $j\in [k]$ and $(\overline{x}_1, ..., \overline{x}_{p})\in (\varrho_l)^{n_j}$.

In the initial instance of $\textsc{CSP}(\bGamma^{\mathcal{B}})$ we have constraints of the following kind: assigned values for a tuple $(v_1,...,v_p)\in r_l$ should be in $\varrho^{\mathcal{B}}_l$, i.e. $\big(o^{h(v_1)}_j, ..., o^{h(v_p)}_j\big)$ should component-wise preserve $\varrho_l$ (for $j\in [k]$). This means that for any $(\overline{x}_1, ..., \overline{x}_{p})\in (\varrho_l)^{n_j}$, we have $\big(o^{h(v_1)}_j(\overline{x}_1), ..., o^{h(v_p)}_j(\overline{x}_{p})\big)\in \varrho_l$. Thus, for any $l\in [s], j\in [k]$, and any $(v_1,...,v_p)\in r_l$, $(\overline{x}_1, ..., \overline{x}_{p})\in (\varrho_l)^{n_j}$, we restrict $\big({\mathsf a}[{v_1,j,\overline{x}_1}], ..., {\mathsf a}[{v_p,j,\overline{x}_{p}}]\big)$ to take its values in $\varrho_l$, and a conjunction of all those constraints is equivalent to the initial constraints of $\textsc{CSP}(\bGamma^{\mathcal{B}})$.

Thus, we described, given an instance of $\textsc{CSP}(\bGamma^{\mathcal{B}})$, how to define a constraint satisfaction problem with relations in constraints taken either from $\{Tr(\mathcal{B})\}$, or from $\Gamma$. By construction, there is a one-to-one correspondence between solutions of the initial CSP and the constructed CSP. \qed
\end{proof}
\else
\fi

Let us denote by $\Gamma^\ast$ the set of all relations over $D$ that can be expressed as a primitive positive formula over $\Gamma$, i.e. as a syntactically correct formula of the form $\exists {\mathbf v} \varrho_1({\mathbf v}_1)\wedge \cdots \wedge \varrho_l({\mathbf v}_l)$ where $\varrho_i\in \Gamma, i\in [l]$ and ${\mathbf v}, {\mathbf v}_i, i\in [l]$ are lists of variables. From Theorem~\ref{simple-red} and~\cite{Jeavons:1998} we conclude:
\begin{corollary}\label{BtoG} If $Tr(\mathcal{B})\in \Gamma^\ast$, then $\textsc{CSP}(\bGamma^{\mathcal{B}})$ is Karp reducible to $\textsc{CSP}(\bGamma)$.
\end{corollary}

\ifTR
Let us translate the premise of Corollary~\ref{BtoG} to the language of polymorphisms.
A well-known fact proved by Geiger~\cite{geiger1968} and by Bodnarchuk-Kalu{\v{z}}nin-Kotov-Romov~\cite{Bodnarchuk} is that $\varrho\in\Gamma^\ast$ if and only if $\varrho$ is preserved by all polymorphisms from $\pol(\Gamma)$. 
To reformulate corollary~\ref{BtoG} using this fact, we need to introduce some definitions.

Let $f$ be an $m$-ary operation on $D$. 
An $m$-ary operation $f^{\mathcal{A}^\sigma_D}$ on $\mathcal{A}^\sigma_D$ is defined by the following rule: $f^{\mathcal{A}^\sigma_{D}}\big(\mathbb{A}_1, \cdots,\mathbb{A}_m\big)=\mathbb{A}$ if and only if for any $j\in [k]$, 
\begin{equation}\label{outside}
o^{\mathbb{A}}_j(x_1,...,x_{n_j})=f\big(o^{\mathbb{A}_1}_j(x_1,...,x_{n_j}), o^{\mathbb{A}_2}_j(x_1,...,x_{n_j}),...,o^{\mathbb{A}_m}_j(x_1,...,x_{n_j})\big).
\end{equation}

We say that $f^{\mathcal{A}^\sigma_{D}}$ preserves $\mathcal{B}$ if $f^{\mathcal{A}^\sigma_{D}}\big(\mathbb{A}_1,...,\mathbb{A}_m\big)\in \mathcal{B}$ whenever $\mathbb{A}_1,...,\mathbb{A}_m \in \mathcal{B}$. 
The following lemma directly follows from the definition of the trace. For completeness, its proof is given in Appendix.
\begin{lemma} \label{bod}  An operation $p: D^m\to D$ is a polymorphism of $Tr(\mathcal{B})$ if and only if $p^{\mathcal{A}^\sigma_D}$ preserves $\mathcal{B}$.
\end{lemma}

From Lemma~\ref{bod} the following corollary is straightforward.
\begin{corollary} If $p^{\mathcal{A}^\sigma_D}$ preserves $\mathcal{B}$ for any $p\in \pol(\Gamma)$, then $\textsc{CSP}(\bGamma^{\mathcal{B}})$ is polynomial-time reducible to $\textsc{CSP}(\bGamma)$.
\end{corollary}

\begin{example}  Let $D=\{0,1\}$ and $\Gamma=\{\{0\}, \{1\}, xy\to z\}$. It is well-known that $\Gamma^\ast = {\rm Inv}(\{\wedge\})$ is a set of Horn predicates. Therefore, if $\wedge^{\mathcal{A}^\sigma_D}$ preserves $\mathcal{B}$, then $\textsc{CSP}(\bGamma^{\mathcal{B}})$ is polynomial-time solvable.
\end{example}
\else
\fi

Theorem~\ref{simple-red} can be slightly strengthened. In order to simplify our notation we will consider a case when the signature $\sigma$ contains only one $n$-ary operation symbol $o$. Recall that $\varrho(\bv) =\big\{ d(\mathbf{v}, \mathbf{y}) | \mathbf{y}\in \varrho\big\}$ (see the definition~\ref{pred-func}).
Thus, $Tr(\mathcal{B}) (\alpha_{n}(1), \cdots,  \alpha_{n}(d^{n}))$ equals
$$
\big\{ \big((\alpha_{n}(1),y_1), \cdots,  (\alpha_{n}(d^{n}),y_{d^{n}})\big) | (y_1, \cdots,  y_{d^{n}})\in Tr(\mathcal{B})\big\}.
$$
Let us denote $\bGamma^n = \big(D^n, \varrho^n_1, \cdots, \varrho^n_s\big)$ (the notation $\varrho^n_i$ is introduced after Definition~\ref{trace-def}).

\begin{theorem}\label{strong-red} $\textsc{CSP}(\bGamma^{\mathcal{B}})$ is polynomial-time Karp reducible to $$\textsc{MCSP}(\Gamma_{\bGamma^n} \cup \{Tr(\mathcal{B}) (\alpha_{n}(1), \cdots,  \alpha_{n}(d^{n}))\}).$$ 
\end{theorem}

\begin{proof}
Let us return to the proof of Theorem~\ref{simple-red} and to the CSP that we constructed in that proof. Since we have only one operation symbol in $\sigma$ we will omit the second index in our variables.
Recall that we had two types of constraints. 
Constraints of the first type require a tuple of variables $$\big({\mathsf a}[{v,\alpha_{n}(1)}], \cdots, {\mathsf a}[{v,\alpha_{n}(d^{n})}]\big)$$ to take its values in $Tr(\mathcal{B})$. Constraints of the second type are as follows: for any $l\in [s]$ and any $(v_1,...,v_p)\in r_l$, $(\overline{x}_1, ..., \overline{x}_{p})\in \varrho_l^{n}$, $\big({\mathsf a}[{v_1,\overline{x}_1}], ..., {\mathsf a}[{v_p,\overline{x}_{p}}]\big)$ should take its values in $\varrho_l$. Thus, in the homomorphism reformulation of CSP, we have to find a homomorphism between a new pair of structures $\bR' = (V',r'_1,\ldots,r'_s, \Xi)$ and $\bGamma\smallfrown\mathcal{B} = \big(D, \varrho_1, \cdots, \varrho_s, Tr(\mathcal{B}) \big)$ where 
\begin{equation*}
\begin{split}
V' = \{{\mathsf a}[v,\alpha_{n}(j)]| v\in V, j\in [d^{n}]\}, \\
r'_l = \{\big({\mathsf a}[{v_1,\overline{x}_1}], ..., {\mathsf a}[{v_p,\overline{x}_{p}}]\big)| (v_1,...,v_p)\in r_l, (\overline{x}_1, ..., \overline{x}_{p})\in \varrho_l^{n}\}, \\
\Xi = \{({\mathsf a}[{v,\alpha_{n}(1)}], \cdots, {\mathsf a}[{v,\alpha_{n}(d^{n})}])| v\in V\}.
\end{split}
\end{equation*}
Let us define $\bGamma_\xi^n = \big(D^n, \varrho^n_1, \cdots, \varrho^n_s, \xi\big)$ where $\xi = \{(\alpha_{n}(1), \cdots, \alpha_{n}(d^{n}))\}$.
By construction a mapping $\delta:V'\to D^n$, where $\delta({\mathsf a}[{v,\overline{x}}]) = \overline{x}$, is a homomorphism from $\bR'$ to $\bGamma_\xi^n$. Thus, we are given a homomorphism to $\bGamma^n_\xi$ and our goal is to find a homomorphism to $\bGamma\smallfrown\mathcal{B}$ which is exactly the definition of $\textsc{CSP}^+_{\bGamma_\xi^n}(\bGamma\smallfrown\mathcal{B})$.

According to Lemma~\ref{prototype}, $\textsc{CSP}^+_{\bGamma_\xi^n}(\bGamma\smallfrown\mathcal{B})$ is equivalent to $\textsc{MCSP}((\bGamma\smallfrown\mathcal{B})_{\bGamma_\xi^n})$. There are 2 types of relations in $(\bGamma\smallfrown\mathcal{B})_{\bGamma_\xi^n}$: those that are in $\Gamma_{\bGamma^n}$ and the relation $Tr(\mathcal{B})\big(\alpha_{n}(1), \cdots, \alpha_{n}(d^{n})\big)$. Therefore, we reduced $\textsc{CSP}(\bGamma^{\mathcal{B}})$ to $\textsc{MCSP}(\Gamma_{\bGamma^n}\cup \{Tr(\mathcal{B}) (\alpha_{n}(1), \cdots,  \alpha_{n}(d^{n}))\})$. \qed
\end{proof}


To formulate a version of Theorem~\ref{strong-red} for a general signature $\sigma$ we need the notion of the disjoint union of relational structures. Given similar structures ${\mathbf T}_i = (A_i, \varrho_{i1}, ..., \varrho_{ik}), i\in [q]$, their disjoint union, denoted $\uplus_{i=1}^q {\mathbf T}_i$, is a structure $(B, \pi_1,..., \pi_k)$ with the domain $B = \cup_{i=1}^q \{i\}\times A_i$ and relations $\pi_j = \cup_{i=1}^q \tau_{ij}$ where $\tau_{ij}$ is a reinterpretation of $\varrho_{ij}$ as a relation over $\{i\}\times A_i$, i.e. $\tau_{ij} = \{\big((i,a_1), ..., (i, a_{\ar(\varrho_{ij})})\big) \mid (a_1, ..., a_{\ar(\varrho_{ij})})\in \varrho_{ij}\}$. We denote $B$ as $\uplus_{i=1}^q A_i$ and $\pi_j$ as $\uplus_{i=1}^q \varrho_{ij}$.
Let us denote $
\bGamma^\sigma = \uplus_{i=1}^s \bGamma^{n_i}$.
Let also $\gamma_1, \cdots,  \gamma_N$ be the ordering of $\cup_{i=1}^k \{i\}\times D^{n_i}$ (the domain of $\bGamma^\sigma$) in which first elements from $\{1\}\times D^{n_1}$ go (in lexicographic order), second $\{2\}\times D^{n_2}$ (in lexicographic order), etc. Then, a generalization of Theorem~\ref{strong-red} is below. Its proof can be found in Appendix.
\begin{theorem}\label{strong-red11} $\textsc{CSP}(\bGamma^{\mathcal{B}})$ is polynomial-time  Karp reducible to $${\textsc MCSP}(\Gamma_{\bGamma^\sigma}\cup \{Tr(\mathcal{B}) (\gamma_1, \cdots,  \gamma_N)\}).$$ 
\end{theorem}

\begin{remark}
$\textsc{MCSP}(\Gamma_{\bGamma^\sigma})$ is tractable, because $\bGamma^\sigma\to\bGamma$ and $\textsc{CSP}^+_{\bGamma^\sigma}(\bGamma)$ is a trivial problem (lemma~\ref{prototype}). Therefore, $\textsc{CSP}(\Gamma_{\bGamma^\sigma})$ is tractable. Moreover, by construction $ \{Tr(\mathcal{B})(\gamma_1, \cdots,  \gamma_N)\}$ is also tractable.
Thus, Theorem~\ref{strong-red11} describes a reduction to an NP-hard language only if the union of those two tractable languages is NP-hard. We conducted some experimental studies with the latter constraint language using the Polyanna software which can be found in Section~\ref{polyanna-exp} of Appendix. 
\end{remark}

\vspace{-10pt}
\section{Conclusions}
\vspace{-5pt}
As examples \ref{tournament}-\ref{TWNU} show, the induced problem 
$\textsc{CSP} (\bGamma^{\mathcal{B}})$ is often easier that the initial $\textsc{CSP} (\bGamma)$. If $\mathcal{B}$ is tractable and one manages to find a homomorphism $\chi:\bR\to \bGamma^{\mathcal{B}}$, then finding $h:\bR\to \bGamma$ can be done efficiently. This inspires the whole family of algorithms based on reducing $\textsc{CSP} (\bGamma)$ to $\textsc{CSP} (\bGamma^{\mathcal{B}})$. This generalizes Green and Cohen's reduction of CSPs to finding appropriate permutations of domains. 

It is an open research problem to generalize the construction of the template $\bGamma^{\mathcal{B}}$ to valued constraint languages (one such example can be found in Section 5 of~\cite{takhanov:LIPIcs:2017:8247}). Practical application of the reduction of $\textsc{CSP} (\bGamma)$ to $\textsc{CSP} (\bGamma^{\mathcal{B}})$ is another topic of future research.
\vspace{-10pt}

%
%
%

%
\bibliographystyle{splncs04}
\newcommand{\noopsort}[1]{}

\newpage
\appendix

\section{Proofs for Section~\ref{sec:relation}}
\begin{proof}[Proof of Lemma~\ref{bod}]
Let us prove that lemma for a simple case when the signature $\sigma$ contains only one $n$-ary operator $o$. Thus,
$$
Tr(\mathcal{B}) = \{\big(o^{{\mathbb A}}(\alpha_{n}(1)), \cdots, o^{{\mathbb A}}(\alpha_{n}(d^{n}))\big) | {\mathbb A}\in\mathcal{B}\}
$$ 
Then, by construction, $p: D^m\to D$ is a polymorphism of $Tr(\mathcal{B})$ if and only if 
for any $\mathbb{A}_1, \cdots,\mathbb{A}_m\in\mathcal{B}$, application of $p$ to the rows of the matrix 
$$
\begin{bmatrix}
o^{{\mathbb A}_1}(\alpha_{n}(1)) & \cdots & o^{{\mathbb A}_m}(\alpha_{n}(1)) \\
\cdots & \cdots & \cdots \\
o^{{\mathbb A}_1}(\alpha_{n}(d^{n}))& \cdots & o^{{\mathbb A}_m}(\alpha_{n}(d^{n})) \\
\end{bmatrix}
$$
gives another element of $Tr(\mathcal{B})$, i.e.
$$
\begin{bmatrix}
o^{{\mathbb A}}(\alpha_{n}(1))  \\
\cdots   \\
o^{{\mathbb A}}(\alpha_{n}(d^{n})) \\
\end{bmatrix}
$$
for an appropriate $\mathbb{A}\in\mathcal{B}$. By construction $\mathbb{A} = p^{\mathcal{A}^\sigma_{D}}\big(\mathbb{A}_1, \cdots,\mathbb{A}_m\big)$, and therefore, $p^{\mathcal{A}^\sigma_{D}}\big(\mathbb{A}_1, \cdots,\mathbb{A}_m\big)\in \mathcal{B}$ whenever $\mathbb{A}_1, \cdots,\mathbb{A}_m\in\mathcal{B}$. All steps in our proof are invertible. \qed
\end{proof}

\begin{proof} [Sketch of proof of Theorem~\ref{strong-red11}]
As in the proof of Theorem~\ref{strong-red} we build $\bR' = (V',r'_1,\ldots,r'_s, \Xi)$ where 
\begin{equation*}
\begin{split}
V' = \{{\mathsf a}[v,i,\alpha_{n_i}(j)]| v\in V, i\in [k], j\in [d^{n_i}]\},\\
r'_l = \{\big({\mathsf a}[{v_1,j,\overline{x}_1}], ..., {\mathsf a}[{v_p,j,\overline{x}_{p}}]\big)|j\in [k],  (v_1,...,v_p)\in r_l, (\overline{x}_1, ..., \overline{x}_{p})\in \varrho_l^{n_j}\}, \\
\Xi = \{({\mathsf a}[{v,1,\alpha_{n_1}(1)}], \cdots, {\mathsf a}[{v,1,\alpha_{n_1}(d^{n_1})}], 
\cdots, {\mathsf a}[{v,k,\alpha_{n_k}(1)}], \\ 
\cdots, {\mathsf a}[{v,k,\alpha_{n_k}(d^{n_k})}])| v\in V\}
\end{split}
\end{equation*}
and define 
$$\bGamma_\xi^n = \big(\uplus_{i=1}^k D^{n_i}, \uplus_{i=1}^k\varrho^{n_i}_1, \cdots, \uplus_{i=1}^k \varrho^{n_i}_s, \xi\big)$$ where $\xi = \{(\gamma_1, \cdots,  \gamma_N)\}$.
By construction a mapping $\delta:V'\to \uplus_{i=1}^k D^{n_i}$, where $\delta({\mathsf a}[{v,j,\overline{x}}]) = (j,\overline{x})$, is a homomorphism from $\bR'$ to $\bGamma_\xi^n$. Thus, we are given a homomorphism to $\bGamma^n_\xi$ and our goal is to find a homomorphism to $\bGamma\smallfrown\mathcal{B}$ which is an instance of $\textsc{CSP}^+_{\bGamma_\xi^n}(\bGamma\smallfrown\mathcal{B})$.

Since $\textsc{CSP}^+_{\bGamma_\xi^n}(\bGamma\smallfrown\mathcal{B})$ and $\textsc{MCSP}(\Gamma_{\bGamma^\sigma}\cup \{Tr(\mathcal{B}) (\gamma_1, \cdots,  \gamma_N)\})$ are equivalent (by Lemma \ref{prototype}), we reduced $\textsc{CSP}(\bGamma^{\mathcal{B}})$ to $\textsc{MCSP}(\Gamma_{\bGamma^\sigma}\cup \{Tr(\mathcal{B}) (\gamma_1, \cdots,  \gamma_N)\})$. \qed
\end{proof}

\subsection{Experiments with Polyanna}\label{polyanna-exp}

\begin{remark}\label{cohen-horn2}
Theorems~\ref{strong-red} and~\ref{strong-red11} do not use all properties of the CSP to which we reduce the initial problem $\textsc{CSP}(\bGamma^{\mathcal{B}})$. Note that in that reduction all variables are partitioned, and any block of the partition ${\mathsf a}[{v,1,\alpha_{n_1}(1)}]$, $\cdots$, ${\mathsf a}[{v,k,\alpha_{n_k}(d^{n_k})}]$ corresponds to a variable $v$ of the initial CSP. Constraints with the relation $Tr(\mathcal{B})(\gamma_1, \cdots,  \gamma_N)$ can be applied only to variables inside of any partition, i.e. such constraints are ``local''. At the same time, constraints with relations from $\Gamma_{\bGamma^\sigma}$ can be applied ``globally''. This property is not used in Theorems~\ref{strong-red} and~\ref{strong-red11}. The latter leads to a non-equivalence of $\textsc{CSP}(\Gamma_{\bGamma^\sigma}\cup \{Tr(\mathcal{B})(\gamma_1, \cdots,  \gamma_N)\})$ and $\textsc{CSP}(\bGamma^{\mathcal{B}})$ in general (though they can be equivalent in special cases).

To verify the last statement, we used Polyanna~\cite{Gault2004}, a specialized software dedicated to the study of polymorphisms for a given language. We experimented with a Boolean domain $D$ and the set ${\mathcal B} = \{(D,\vee), (D, \wedge)\}$. According to Example~\ref{ExCohen}, $\textsc{CSP}(\bGamma^{\mathcal B})$ is tractable. It also follows from the fact that the domain permutation reduction to a max-closed constraint language is always tractable for Boolean domains~\cite{GREEN20081094}.
In our experiments with Polyanna, when the chosen language $\Gamma$ was NP-hard, $\Gamma_{\bGamma^n}\cup \{Tr(\mathcal{B})(\alpha_{n}(1), \cdots, \alpha_{n}(d^{n}))\}$ was also NP-hard (i.e. harder than $\textsc{CSP}(\bGamma^{\mathcal B})$). 
\end{remark}

\begin{remark}
We also conducted experimental studies in the Boolean case when $\bGamma$ is fixed, and ${\mathcal B}$ is varied. Note that Polyanna could not answer an overwhelming majority of cases because of the computational hardness of finding WNU polymorphism for languages with large domains. That is why we do not infer conjectures from the following findings. 

For $\bGamma = (D; \varrho_{\rm 1-in-3})$, where $\varrho_{\rm 1-in-3}=\{(0,0,1), (0,1,0), (1,0,0)\}$, and ${\mathcal B} \subseteq \{{\mathbb A} = (D,b)| b: D^2\rightarrow D\}$, Polyanna identifies $\Gamma_{\bGamma^n}\cup \{Tr(\mathcal{B})(\alpha_{n}(1), \cdots,  \alpha_{n}(d^{n}))\}$ as tractable if and only if either $(D, b(x,y)=x)\in {\mathcal B} $ or $(D, b(x,y)=y)\in {\mathcal B}$.

For $\bGamma = (D; \varrho_{\rm NAE})$, where $$\varrho_{\rm NAE}=\{(0,0,1), (0,1,0), (1,0,0), (1,1,0), (1,0,1), (0,1,1)\},$$ and ${\mathcal B} \subseteq \{{\mathbb A} = (D,b)| b: D^2\rightarrow D\}$, Polyanna identifies $\Gamma_{\bGamma^n}\cup \{Tr(\mathcal{B}) (\alpha_{n}(1), ..., \alpha_{n}(d^{n}))\}$ as tractable if  and only if $(D, b(x,y)=x)\in {\mathcal B} $ or $(D, b(x,y)=y)\in {\mathcal B} $ or $(D, b(x,y)=\neg{x})\in {\mathcal B} $ or $(D, b(x,y)=\neg{y})\in {\mathcal B} $. Here $\neg{x}$ is the negation of $x$.

In all our experiments with NP-hard $\Gamma$, the language $\Gamma_{\bGamma^n}\cup \{Tr(\mathcal{B})(\alpha_{n}(1), ..., \alpha_{n}(d^{n}))\}$ was identified by Polyanna as tractable only when $\bGamma^{\mathcal B}$ was constant preserving.
\end{remark}

\section{Examples}\label{nice-examples}
Let us now list a number of specific examples that demonstrate that $\textsc{CSP}(\bGamma^{\mathcal{B}})$ can be easier than $\textsc{CSP}(\bGamma)$ and could motivate a number of non-trivial algorithms for $\textsc{CSP}(\bGamma)$. 

Below, in subsections \ref{tournament}-\ref{TWNU}, we give examples of $\mathcal{B}$ for which $\textsc{CSP}(\bGamma^{\mathcal{B}})$ is tractable for any (possibly NP-hard) $\bGamma$. In all examples below, $\varrho$ is an arbitrary relation over $D$, and $\bGamma$ is an arbitrary relational structure over $D$.

\subsection{Tournament pairs}\label{tournament} All the argument of the previous example remain unchanged if we set
$$
\mathcal{B}_1 = \big\{(D, a, b)| \{a(x,y),b(x,y)\}= \{x,y\}, a(x,y)=a(y,x), b(x,y)=b(y,x)\big\}.
$$
These are algebras of so called tournament pairs. They appear in the context of Minimum Cost Homomorphism problems~\cite{Takhanov10adichotomy,cohen08:generalising}. They can also be understood using the following formulation. For any $(D, a, b)\in \mathcal{B}_1$, if we choose any two distinct elements $e_0, e_1\in D$ and identify $e_0$ with 0, and $e_1$ with 1, then a pair of operations $a,b$ restricted to $\{0,1\}$ should be either $(\vee, \wedge)$ or $(\wedge, \vee)$. In other words, $a|_{\{0,1\}}$ is either $\vee$ or $\wedge$, and $b|_{\{0,1\}}$ is a Boolean function dual to $a|_{\{0,1\}}$.
Thus, $|\mathcal{B}_1|=2^{|D| \choose 2}$.

We now define a majority polymorphism $m_1$ by the following rule:
\begin{align*}
m_1((D, a_1, b_1),(D, a_2, b_2),(D, a_3, b_3)) = (D, a, b) \Leftrightarrow \\
a(x,y)=a_1(a_2(x,y),a_3(x,y)), b(x,y)=b_1(b_2(x,y),b_3(x,y)).
\end{align*}
Let us check that if $(D, a_1, b_1),(D, a_2, b_2),(D, a_3, b_3)\in \mathcal{B}_1$ then the resulting algebra $(D, a, b)$ is also in $\mathcal{B}_1$. For $D=\{0,1\}$, the latter can be seen from a well-known fact about Boolean functions: $x\vee y, x\wedge y$ are dual, and, if in any monotone logical expression $T$ that can be build from $\vee , \wedge $ we interchange $\vee$ and $\wedge $, then we obtain the dual of $T$. Thus, $a(x,y)$ and $b(x,y)$ are dual to each other (as $a_i$ and $b_i$), conservative and commutative. Therefore, $(D, a, b)\in \mathcal{B}_1$. For general $D$, we only have to refer to the previous case: for any two distinct $e_0, e_1\in D$ if we identify $e_0$ with 0, and $e_1$ with 1, then $b|_{\{0,1\}}$ becomes dual to $a|_{\{0,1\}}$.

Identically to Example~\ref{ExCohen}, it can be shown that $m_1$ is a majority polymorphism of $\varrho^{\mathcal{B}}$ for any $\varrho\subseteq D^k$. Therefore, $\textsc{CSP}(\bGamma^{\mathcal{B}})$ is tractable  for any $\bGamma$.

\subsection{Commutative idempotent binary operations} A direct extension of the set of algebras of example~\ref{ExCohen} is the following set of algebras
$$
\mathcal{B}_{\rm com} = \big\{(D, a)|a:D^2\to D, a(x,y)=a(y,x), a(x,x)=x\big\}.
$$
Let us introduce an idempotent binary operation $\Diamond:\mathcal{B}_{\rm com}^2\to \mathcal{B}_{\rm com} $ by:
$$
\Diamond((D, b_1),(D, b_2)) = (D, b) \Leftrightarrow b(x,y)=b_1(b_1(x,y),b_2(x,y)).
$$
Also, let us define a unary operation $\house: \mathcal{B}_{\rm com}\to \mathcal{B}_{\rm com}$:
$$
\house ((D, b_1)) = (D, b) \Leftrightarrow b(x,y)=b_1(b_1(b_1(x,y),x),b_1(b_1(x,y),y)).
$$

\begin{figure}
  \centering
  \begin{center}
    \includegraphics[width=0.35\textwidth,trim={5cm 5cm 6cm 3cm},clip]{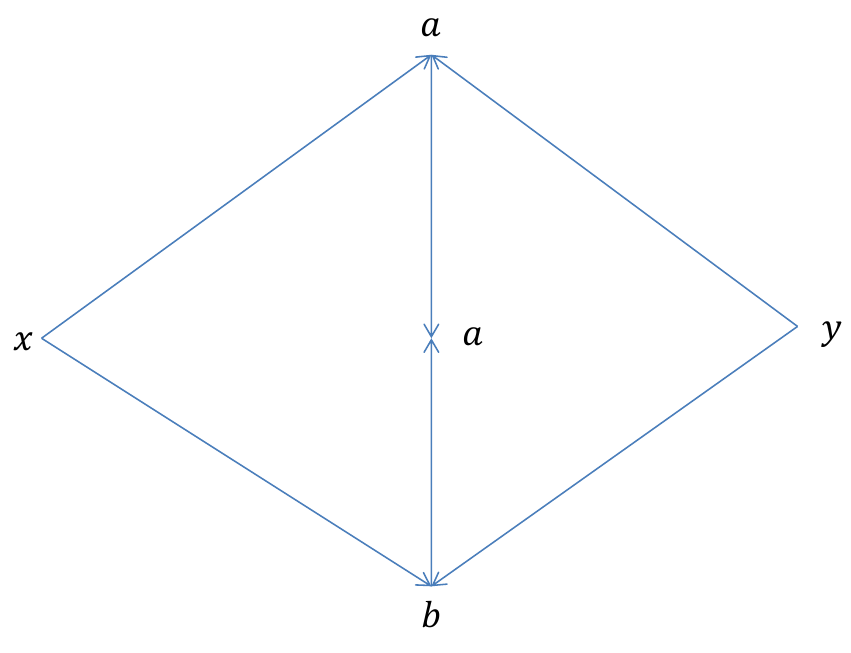} \includegraphics[width=0.2\textwidth,trim={6cm 2cm 8cm 1cm},clip]{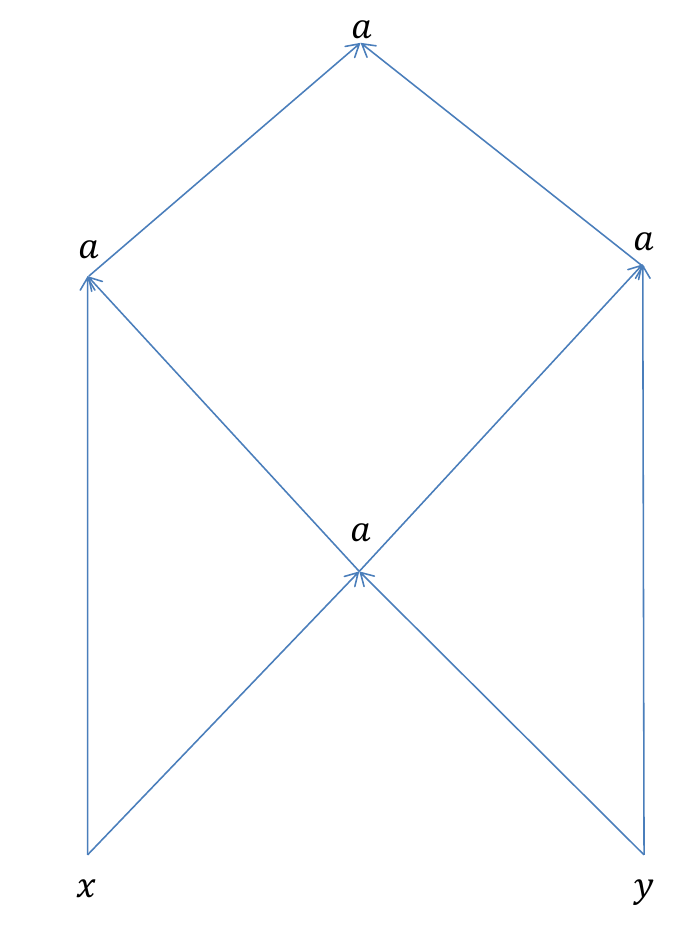}
  \end{center}
  \caption{Circuits representing $\Diamond((D, a),(D, b))$ and $\house ((D, a))$.}
\end{figure}

It is well-known that for any unary $\phi\in \pol(\bGamma^{\mathcal{B}_{\rm com}})$ such that $\phi\circ\phi=\phi$, $\textsc{CSP}(\bGamma^{\mathcal{B}_{\rm com}})$ can be reduced to $\textsc{CSP}\big( \{\varrho\cap \phi(\mathcal{B}_{\rm com})^{\ar(\varrho)}| \varrho\in\Gamma^{\mathcal{B}_{\rm com}}\}\big)$. Thus, the tractability of $\Gamma^{\mathcal{B}_{\rm com}}$ is equivalent to the tractability of $\{\varrho\cap \phi(\mathcal{B}_{\rm com})^{\ar(\varrho)}| \varrho\in\Gamma^{\mathcal{B}_{\rm com}}\}$, and therefore, is defined by polymorphisms of $\{\varrho\cap \phi(\mathcal{B}_{\rm com})^{\ar(\varrho)}| \varrho\in\Gamma^{\mathcal{B}_{\rm com}}\}$ on the sub-domain $\phi(\mathcal{B}_{\rm com})$.

Since $\house\in \pol(\bGamma^{\mathcal{B}_{\rm com}})$ for any $\Gamma$, then for any $n\in {\mathbb N}$, $ \house^n(x) = \house(\house(\cdots \house(x)\cdots))$ is a polymorphism of $\bGamma^{\mathcal{B}_{\rm com}}$. 
\begin{theorem}
Let $|D|\leq 5$ and let $n\in {\mathbb N}$ be the smallest number such that $\house^n\circ \house^n=\house^n$. Then $B:(\house^n(\mathcal{B}_{\rm com}))^2\to \house^n(\mathcal{B}_{\rm com})$, defined by $B(x,y) = \house^n(\Diamond(x,y))$, is idempotent and commutative.
\end{theorem}

A computer proof of that theorem is based on checking $B(x,y)=B(y,x)$ for any pair $x,y\in \house^n(\mathcal{B}_{\rm com})$ using a Python 3 script (see \url{https://github.com/k-nic/constraints}). 

Since $\Diamond, \house \in \pol(\Gamma^{\mathcal{B}_{\rm com}})$, then $B'(x,y) = \house^n(\Diamond(\house^n(x),\house^n(y)))$ is a polymorphism of $\Gamma^{\mathcal{B}_{\rm com}}$ and $B(x,y)=B'(x,y), (x,y)\in \house^n(\mathcal{B}_{\rm com})^2$ is a polymorphism of $\{\varrho\cap \house^n(\mathcal{B}_{\rm com})^{\ar(\varrho)}| \varrho\in\Gamma^{\mathcal{B}_{\rm com}}\}$. Since $B$ is idempotent and binary, we obtain that $\{\varrho\cap \house^n(\mathcal{B}_{\rm com})^{\ar(\varrho)}| \varrho\in\Gamma^{\mathcal{B}_{\rm com}}\}$ is tractable. Thus,
\begin{theorem}
If $|D|\leq 5$, then $\bGamma^{\mathcal{B}_{\rm com}}$ is tractable for any $\Gamma$.
\end{theorem}

We conjecture that $\bGamma^{\mathcal{B}_{\rm com}}$ is tractable for any finite $D$, though for $|D|\geq 6$, the cardinality $|\mathcal{B}_{\rm com}|=|D|^{{|D| \choose 2}}\geq 470184984576$ is so large that $\textsc{CSP}(\bGamma^{\mathcal{B}_{\rm com}})$ loses  any practical importance.

\ifCOM

Let $D=\{0,1,2\} = {\mathbb Z}/3{\mathbb Z}$, $\mathcal{B} = \big\{(D, w)| w(x,y)=w(y,x), w(x,x)=x\big\}$ and 
$$\mathcal{B}_i = \big\{(D, w)| w(x,y)=w(y,x), w(x,x)=x,w(i,x)=i\big\} \subseteq \mathcal{B}
$$
for $i\in D$. Let us set $\mathcal{B}_{012} = \mathcal{B}_0\cup \mathcal{B}_1\cup \mathcal{B}_2$. 
Any binary commutative operation $b$ on $D$ is defined by its values $b(0,1), b(0,2), b(1,2)$, therefore $|\mathcal{B}|=3^3=27$. For any $\pi:D\to D$, $i\circ^{\pi} j = \pi(i+j)$ is a binary commutative operation ($+$ is a summation modulo 3). Since $|D^D| = 27$, then $\pi\to i\circ^{\pi} j$ is a one-to-one mapping between $D^D$ and $\mathcal{B}$.

As in example~\ref{ExCohen}, let us introduce an idempotent ternary operation $m:\mathcal{B} ^3\to \mathcal{B} $ by:
$$
m((D, b_1),(D, b_2),(D, b_3)) = (D, b) \Leftrightarrow b(x,y)=b_1(b_2(x,y),b_3(x,y)).
$$
Note that $m(y,x,x)=x$ and the image $m(\mathcal{B}_{012},\mathcal{B}_{012},\mathcal{B}_{012})$ is not a subset of $\mathcal{B}_{012}$. Let us now consider an idempotent operation $M(x,y,z) = m(x,y,m(y,x,z))$ and prove that $M(y,x,x)=M(x,y,x)=M(x,x,y)$ for any $x,y\in \mathcal{B}_{012}$. 

W.l.o.g. we can assume $x = (D,a)\in \mathcal{B}_0$, i.e. $a(0,u)=0$. If we denote $y=(D,b)$, the condition $\mathop\forall\limits_{x,y\in \mathcal{B}_{012}}M(y,x,x)=M(x,y,x)=M(x,x,y)$ is equivalent to
\begin{equation*}
\begin{split}
b(a(u,v),a(b(u,v), a(u,v))) = a(b(u,v),b(a(u,v), a(u,v))) = a(b(u,v),a(u,v)) =\\a(a(u,v),a(a(u,v), b(u,v))).
\end{split}
\end{equation*}
If $0\in \{u,v\}$, then $a(u,v)=0$, and the latter expressions are all equal to 0. It remains to check the case $u=1, v=2$. If $b(1,2)=0$, then $b\notin \mathcal{B}_{1}\cup \mathcal{B}_{2}$, i.e. $b\in \mathcal{B}_{0}$. Therefore, expressions are again equal to 0. If $b(1,2)=1$, then the latter identity reads as
\begin{equation*}
\begin{split}
b(a(1,2),a(1, a(1,2))) = a(1,a(1,2)) =a(a(1,2),a(a(1,2), 1)).
\end{split}
\end{equation*}
If $a(1,2)=i$, then $a(1,i)=i$. Therefore, $b(i,a(1,i)) = a(1,i) =a(i,a(i, 1))=i$ and the latter identities hold. The case $b(1,2)=2$ is treated analogously. Thus, we proved $\mathop\forall\limits_{x,y\in \mathcal{B}_{012}}M(y,x,x)=M(x,y,x)=M(x,x,y)$.

A direct computation gives:
\begin{equation*}
\begin{split}
u(\mathcal{B}) = \mathcal{B}_{012}\cup \{(D,b)| b(u,v)=\pi(u+v), \pi\in S(3)\} \\
u(u(\mathcal{B})) = u(\mathcal{B}) \\
u(M(\mathcal{B}_{012},\mathcal{B}_{012},\mathcal{B}_{012}))\subseteq \mathcal{B}_{012}
\end{split}
\end{equation*}
where $S(3) = \{\pi:D\to D| \pi{\rm \,\,is\,\,a\,\,permutation}\}$. Since the restriction $u|_{u(\mathcal{B})}: u(\mathcal{B})\to u(\mathcal{B})$ is a permutation, then there is $n\in {\mathbb N}$ such that $u^n(x)=x$. Thus, $N(x,y,z) = u^n(M(x,y,z))$ satisfies:
\begin{equation*}
\begin{split}
N(\mathcal{B}_{012},\mathcal{B}_{012},\mathcal{B}_{012})\subseteq \mathcal{B}_{012} (???) \\
\mathop\forall\limits_{x,y\in \mathcal{B}_{012}} N(y,x,x)=N(x,y,x)=N(x,x,y) \\
\mathop\forall\limits_{x\in \mathcal{B}_{012}}N(x,x,x)=x.
\end{split}
\end{equation*}
In other words $N\in {\rm Pol}(\bGamma^{\mathcal{B}_{012}})$ for any constraint language $\Gamma$ and $N$ is a weak majority operation.
\else
\fi
\subsection{Near unanimity operations}
Let $n\in {\mathbb N}, n\geq 3$ be fixed and
\begin{equation*}
\begin{split}
\mathcal{B}_{\rm nu} = \big\{(D, w)| w:D^n\to D, w(x,\cdots,x,y)=w(x,\cdots,x,y,x)=\\
\cdots = w(y,x,\cdots,x) = x\big\}
\end{split}
\end{equation*}
be a set of near unanimity operations of arity $n$.

An $n+1$-ary operation $T: \mathcal{B}^{n+1}_{\rm nu}\rightarrow \mathcal{B}_{\rm nu}$ defined by
\begin{equation*}
\begin{split}
T((D, w_1), (D, w_2), \cdots, (D, w_{n+1})) = (D, w) \Leftrightarrow \\
w(x_1,x_2,\cdots, x_n) = \\
w_{n+1}(w_1(x_1,x_2,\cdots, x_n),w_2(x_1,x_2,\cdots, x_n),\cdots, w_n(x_1,x_2,\cdots, x_n)) 
\end{split}
\end{equation*}
is itself a near unanimity. 

Indeed, if $w_1=w_2=\cdots=w_n=w\in \mathcal{B}_{\rm nu}$ and $w_{n+1}=v\in \mathcal{B}_{\rm nu}$, then
$$
v(w(\cdots),w(\cdots),\cdots, w(\cdots)) = w(\cdots)
$$
because $v$ is idempotent. If $w_2=\cdots=w_n=w_{n+1} = w$ and $w_{1}=v$, then 
$$
w(v(\cdots),w(\cdots),\cdots, w(\cdots)) = w(\cdots)
$$
because $w$ is a near unanimity. Analogously, the case $w_i=v$ and $w_j=w, j\ne i$ is treated.

It is easy see that $T$ is a polymorphism of $\bGamma^{\mathcal{B}_{\rm nu}}$ for any $\Gamma$. Thus, from the tractability of algebras with a near unanimity term~\cite{JEAVONS1998251} we have
\begin{theorem}
$\bGamma^{\mathcal{B}_{\rm nu}}$ is of bounded width for any $\Gamma$.
\end{theorem}

\subsection{Maltsev operations}
Let
$$
\mathcal{B}_{\rm M} = \big\{(D, w)|w:D^3\to D, w(x,x,y)=w(y,x,x)=y\big\}
$$
be a set of Maltsev operations on $D$.
Let us define a 4-ary operation $S: \mathcal{B}^4_{\rm M}\rightarrow \mathcal{B}_{\rm M}$ by the following rule:
\begin{align*}
S((D, w_1), (D, w_2), (D, w_3), (D, w_4)) = (D, w) \Leftrightarrow \\
w(x,y,z) = w_1(w_2(x,y,z),w_3(x,y,z),w_4(x,y,z)).
\end{align*}
By construction, $S$ is idempotent and $M(x,y,z) = S(x,x,y,z)$ satifies
$$
M(x,y,y)=S(x,x,y,y) = x
$$
and
$$
M(y,y,x)=S(y,y,y,x) = x,
$$
i.e. $M$ is itself a Maltsev operation.
Again, by construction, both $S$ and $M$ are polymorphisms of $\bGamma^{\mathcal{B}_{\rm M}}$ for any $\Gamma$. Therefore, from the tractability of Maltsev constraints~\cite{doi:10.1137/050628957} we have
\begin{theorem}
$\bGamma^{\mathcal{B}_{\rm M}}$ is tractable for any $\Gamma$.
\end{theorem}

\subsection{Ternary weak near unanimity operations} \label{TWNU}
Let
$$
\mathcal{B}_{\rm wnu} = \big\{(D, w)| w(x,x,y)=w(x,y,x)=w(y,x,x), w(x,x,x)=x\big\}
$$
for a general $D$. These are algebras of ternary idempotent weak near-unanimity operations.

Let us define a 4-ary operation $t: \mathcal{B}^4_{\rm wnu}\rightarrow \mathcal{B}_{\rm wnu}$ by the following rule:
\begin{align*}
t((D, w_1), (D, w_2), (D, w_3), (D, w_4)) = (D, w) \Leftrightarrow \\
w(x,y,z) = w_1(w_2(x,y,z),w_3(x,y,z),w_4(x,y,z)).
\end{align*}
The fact that the resulting operation $w$ is an idempotent weak near-unanimity operation is trivial. 

\begin{theorem} For $|D|=2{\rm \,\,or\,\,}3$, an operation $r:\mathcal{B}^4_{\rm wnu}\rightarrow \mathcal{B}_{\rm wnu}$, defined by
$$
r({\mathbb A}_1,{\mathbb A}_2,{\mathbb A}_3, {\mathbb A}_4) = \phi(t({\mathbb A}_1,{\mathbb A}_2,{\mathbb A}_3, {\mathbb A}_4), t({\mathbb A}_4,{\mathbb A}_3,{\mathbb A}_2, {\mathbb A}_1))
$$
where $\phi({\mathbb A},{\mathbb B}) = t({\mathbb A},{\mathbb B},{\mathbb A},{\mathbb A})$, is a weak near-unanimity operation.
\end{theorem}
\begin{proof} Our proof of the statement for $|D|=3$ is based on a brute-force checking  $r({\mathbb A},{\mathbb B},{\mathbb B}, {\mathbb B}) = r({\mathbb B},{\mathbb A}, {\mathbb B}, {\mathbb B})=r({\mathbb B},{\mathbb B}, {\mathbb A}, {\mathbb B})=r({\mathbb B},{\mathbb B}, {\mathbb B}, {\mathbb A})$  over all pairs $({\mathbb A}, {\mathbb B})\in \mathcal{B}_{\rm wnu}^2$, using a computer program written on Python 3 (see \url{https://github.com/k-nic/constraints}). The cardinality $|\mathcal{B}_{\rm wnu}| = |D|^{|D|(|D|-1)(|D|-2)+|D|(|D|-1)}$ grows too rapidly, therefore such a simple approach for $|D|=4$ is not possible.

For $D=\{0,1\}$, a small size proof can be done. Let us demonstrate it.
By construction $\mathcal{B}_{\rm wnu}$ consists of 4 algebras with the following operations: $x\wedge y \wedge z$, $x\vee y \vee z$, $\textsc{Mj}(x,y,z)$, $\textsc{Mn}(x,y,z)$ where $\textsc{Mj}(x,y,z) = [x+y+z\geq 2]$ and $\textsc{Mn}(x,y,z) = [x+y+z\equiv 1(\rm{mod }\,\,2)]$ ($[\cdot]$ is the Iverson bracket). Let us denote $\mathbf{1} = (D, x\wedge y \wedge z)$, $\mathbf{2} = (D, x\vee y \vee z)$, $\mathbf{3} = (D, \textsc{Mj})$ and $\mathbf{4} = (D, \textsc{Mn})$. Thus,
$$
\mathcal{B}_{\rm wnu} = \big\{\mathbf{1}, \mathbf{2}, \mathbf{3}, \mathbf{4}\big\}.
$$

The operation $t:\mathcal{B}^4_{\rm wnu}\rightarrow \mathcal{B}_{\rm wnu}$ itself is idempotent but is not a weak near-unanimity operation. 
In the function's table given below we list input tuples $(\mathbb{A}_1, \mathbb{A}_2, \mathbb{A}_3, \mathbb{A}_4)$ in which some 3 arguments out of 4 are the same. We omitted all tuples that can be obtained from the listed tuples by permuting ${\mathbb A}_2,{\mathbb A}_3, {\mathbb A}_4$ because the output of $t$ is preserved under such permutations. Also, we omitted tuples $({\mathbb A},{\mathbb A},{\mathbb A},{\mathbb A})$.
\begin{center}
    \begin{tabular}{ | l | l | l | l | l | l | l | l | l | l | l | l | l | l | l | l | l | l | l | l | l | l | l |}
    \hline
    $\mathbb{A}_1$ & $\mathbb{A}_2$ & $\mathbb{A}_3$ & $\mathbb{A}_4$ & $t$ & | & $\mathbb{A}_1$ & $\mathbb{A}_2$ & $\mathbb{A}_3$ & $\mathbb{A}_4$ & $t$ \\
    \hline
    $\mathbf{1}$ & $\mathbf{2}$ & $\mathbf{2}$ & $\mathbf{2}$ & $\mathbf{2}$ & & $\mathbf{2}$ & $\mathbf{1}$ & $\mathbf{1}$ & $\mathbf{1}$ & $\mathbf{1}$\\
    \hline
    $\mathbf{2}$ & $\mathbf{1}$ & $\mathbf{2}$ & $\mathbf{2}$ & $\mathbf{2}$ & & $\mathbf{1}$ & $\mathbf{2}$ & $\mathbf{1}$ & $\mathbf{1}$ & $\mathbf{1}$\\
    \hline
    $\mathbf{1}$ & $\mathbf{3}$ & $\mathbf{3}$ & $\mathbf{3}$ & $\mathbf{3}$ & & $\mathbf{3}$ & $\mathbf{1}$ & $\mathbf{1}$ & $\mathbf{1}$ & $\mathbf{1}$\\
    \hline
    $\mathbf{3}$ & $\mathbf{1}$ & $\mathbf{3}$ & $\mathbf{3}$ & $\mathbf{3}$ & & $\mathbf{1}$ & $\mathbf{3}$ & $\mathbf{1}$ & $\mathbf{1}$ & $\mathbf{1}$\\
    \hline
    $\mathbf{1}$ & $\mathbf{4}$ & $\mathbf{4}$ & $\mathbf{4}$ & \cellcolor{green!25}$\mathbf{4}$ & & $\mathbf{4}$ & $\mathbf{1}$ & $\mathbf{1}$ & $\mathbf{1}$ & $\mathbf{1}$\\
    \hline
    $\mathbf{4}$ & $\mathbf{1}$ & $\mathbf{4}$ & $\mathbf{4}$ & \cellcolor{green!25}$\mathbf{1}$ & & $\mathbf{1}$ & $\mathbf{4}$ & $\mathbf{1}$ & $\mathbf{1}$ & $\mathbf{1}$\\
    \hline
$\mathbf{2}$ & $\mathbf{3}$ & $\mathbf{3}$ & $\mathbf{3}$ & $\mathbf{3}$ & & $\mathbf{3}$ & $\mathbf{2}$ & $\mathbf{2}$ & $\mathbf{2}$ & $\mathbf{2}$\\
    \hline
$\mathbf{3}$ & $\mathbf{2}$ & $\mathbf{3}$ & $\mathbf{3}$ & $\mathbf{3}$ & & $\mathbf{2}$ & $\mathbf{3}$ & $\mathbf{2}$ & $\mathbf{2}$ & $\mathbf{2}$\\
    \hline
$\mathbf{2}$ & $\mathbf{4}$ & $\mathbf{4}$ & $\mathbf{4}$ & \cellcolor{yellow!25}$\mathbf{4}$ & & $\mathbf{4}$ & $\mathbf{2}$ & $\mathbf{2}$ & $\mathbf{2}$ & $\mathbf{2}$\\
    \hline
$\mathbf{4}$ & $\mathbf{2}$ & $\mathbf{4}$ & $\mathbf{4}$ & \cellcolor{yellow!25}$\mathbf{2}$ & & $\mathbf{2}$ & $\mathbf{4}$ & $\mathbf{2}$ & $\mathbf{2}$ & $\mathbf{2}$\\
    \hline
$\mathbf{3}$ & $\mathbf{4}$ & $\mathbf{4}$ & $\mathbf{4}$ & \cellcolor{blue!25}$\mathbf{4}$ & & $\mathbf{4}$ & $\mathbf{3}$ & $\mathbf{3}$ & $\mathbf{3}$ & $\mathbf{3}$\\
    \hline
$\mathbf{4}$ & $\mathbf{3}$ & $\mathbf{4}$ & $\mathbf{4}$ & \cellcolor{blue!25}$\mathbf{3}$ & & $\mathbf{3}$ & $\mathbf{4}$ & $\mathbf{3}$ & $\mathbf{3}$ & $\mathbf{3}$\\
    \hline
    \end{tabular}
\end{center}

By construction, for $\{{\mathbb A},{\mathbb B}\} = \{\mathbf{1}, \mathbf{4}\},  \{\mathbf{2}, \mathbf{4}\},  \{\mathbf{3}, \mathbf{4}\}$, we have $t({\mathbb A}, {\mathbb B}, {\mathbb A},  {\mathbb A}) \ne t({\mathbb B}, {\mathbb A}, {\mathbb A}, {\mathbb A})$ (see colored cells in the table). For any other pair of algebras $\{{\mathbb A},{\mathbb B}\} \subseteq \mathcal{B}_{\rm wnu}$, the table shows that
\begin{equation}\label{WNU-prop}
t({\mathbb B}, {\mathbb A}, {\mathbb A}, {\mathbb A})=t({\mathbb A}, {\mathbb B}, {\mathbb A}, {\mathbb A}) = t({\mathbb A}, {\mathbb A}, {\mathbb B}, {\mathbb A})=t({\mathbb A}, {\mathbb A}, {\mathbb A}, {\mathbb B}).
\end{equation}
It is also easy to see from the table that $\phi({\mathbb A},{\mathbb B}) = t({\mathbb A},{\mathbb B},{\mathbb A},{\mathbb A})$ has the following property:
$$
\phi({\mathbb A},{\mathbb B}) = \phi({\mathbb B},{\mathbb A})\,\,\,\forall \{{\mathbb A},{\mathbb B}\} = \{\mathbf{1}, \mathbf{4}\},  \{\mathbf{2}, \mathbf{4}\},  \{\mathbf{3}, \mathbf{4}\}.
$$
Thus, it is easy to check that the following operation has the property~\eqref{WNU-prop} for all pairs $\{{\mathbb A},{\mathbb B}\}$:
$$
r({\mathbb A}_1,{\mathbb A}_2,{\mathbb A}_3, {\mathbb A}_4) = \phi(t({\mathbb A}_1,{\mathbb A}_2,{\mathbb A}_3, {\mathbb A}_4), t({\mathbb A}_4,{\mathbb A}_3,{\mathbb A}_2, {\mathbb A}_1)).
$$

\end{proof}

The fact that $t$ (and, consequently, $r$) is a polymorphism of any $\varrho^{\mathcal{B}_{\rm wnu}}$ can be proved analogously to the proof in Example~\ref{ExCohen}. Since $r$ is a 4-ary weak near-unanimity operation, then from the algebraic version of the dichotomy theorem~\cite{bulatov05:classifying} we know that $\textsc{CSP}(\bGamma^{\mathcal{B}_{\rm wnu}})$ is tractable.

\subsection{Siggers pairs}
\begin{definition}
An algebra $(D; g, s)$ is called a Siggers pair on a domain $D$ if $g$ is a unary operation on
$D$ satisfying $g\circ g = g$ and $s$ is a 4-ary operation that satisfies:
\begin{equation*}
\begin{split}
\forall x,y,z,t\in g(D)\,\,\, s(x,y,z,t)\in g(D),\, s(x,x,x,x)=x \\
\forall r,a,e\in g(D)\,\,\, s(r,a,r,e)=s(a,r,e,a).
\end{split}
\end{equation*}
\end{definition}

It is well-known that $\Gamma$ is tractable if and only if $\Gamma\subseteq {\rm Inv}(\{g,s\})$ for some Siggers pair $(D; g, s)$~\cite{Siggers,Kearnes2014}. Let $\mathcal{B}_{\rm S}$ denote a set of all Siggers pairs on $D$. The relational structure $\bGamma^{\mathcal{B}_{\rm S}}$ has already been studied in~\cite{takhanov:LIPIcs:2017:8247}\footnote{There it was denoted by $\bGamma'$}. The structure has the following remarkable properties: a) for any $\bR$, $\Gamma_{\bR}$ is tractable if and only if $\bR\to \bGamma^{\mathcal{B}_{\rm S}}$ (Lemma 26 in~\cite{takhanov:LIPIcs:2017:8247}), b) $\textsc{CSP}(\bGamma)$ is polynomially Turing reducible to $\textsc{CSP}(\bGamma^{\mathcal{B}_{\rm S}})$ (Theorem 28 in~\cite{takhanov:LIPIcs:2017:8247}), c) from the latter we conclude that if $\bGamma$ is NP-hard, then $\bGamma^{\mathcal{B}_{\rm S}}$ is also NP-hard. 
This example is the only example in the current section in which $\bGamma^{\mathcal{B}_{\rm S}}$ is not polynomially solvable for any $\bGamma$.
\end{document}